\documentclass[11pt,a4paper]{article}
\usepackage{fullpage,graphicx,amssymb,amsmath,xspace,amsfonts,tabularx,rotating,url}

\newcommand{\Reals}{\mathbb{R}}

\newcommand{\Naturals}{\mathbb{N}}
\def\andfrac#1/#2{%
   \leavevmode\kern.1em
   \raise.5ex\hbox{\the\scriptfont0 #1}\kern-.1em
   /\kern-.15em\lower.25ex\hbox{\the\scriptfont0 #2}}
\def\WLMax{\ensuremath{\mathrm{WL}_\infty}\xspace}
\def\WLEuc{\ensuremath{\mathrm{WL}_2}\xspace}
\def\WLMan{\ensuremath{\mathrm{WL}_1}\xspace}
\def\WS{\ensuremath{\mathrm{WS}}\xspace}
\def\WBV{\ensuremath{\mathrm{WBV}}\xspace}
\def\WBS{\ensuremath{\mathrm{WBS}}\xspace}
\def\tail{\ensuremath{\mathrm{tail}}\xspace}
\def\head{\ensuremath{\mathrm{head}}\xspace}
\def\bbox{\ensuremath{\mathrm{bbox}}\xspace}

\newtheorem{lemma}{Lemma}
\newtheorem{theorem}{Theorem}
\newenvironment{proof}{Proof:}{\qed}
\def\squareforqed{\hbox{\rlap{$\sqcap$}$\sqcup$}}
\def\qed{\ifmmode\squareforqed\else{\unskip\nobreak\hfil
\penalty50\hskip1em\null\nobreak\hfil\squareforqed
\parfillskip=0pt\finalhyphendemerits=0\endgraf}\fi}

\long\def\skipsection#1{}

\begin{document}

\title{An inventory of three-dimensional Hilbert space-filling curves}
\author{%
Herman~Haverkort\thanks{Dept.\ of Computer Science, Eindhoven University of Technology, the Netherlands, cs.herman@haverkort.net}
}
\date{September 13, 2011}
\maketitle

\begin{abstract}
Hilbert's two-dimensional space-filling curve is appreciated for its good locality properties for many applications. However, it is not clear what is the best way to generalize this curve to filling higher-dimensional spaces. We argue that the properties that make Hilbert's curve unique in two dimensions, are shared by 10\,694\,807 structurally different space-filling curves in three dimensions. These include several curves that have, in some sense, better locality properties than any generalized Hilbert curve that has been considered in the literature before.
\textbf{I recommend to first read Appendix~\ref{apx:update}: Update and erratum.}
\end{abstract}

\section{Introduction}
A space-filling curve in $d$ dimensions is a continuous, surjective mapping from $\Reals$ to $\Reals^d$.
In the late 19th century Peano~\cite{Peano} described such mappings for $d = 2$ and $d = 3$. Since then, quite a number of space-filling curves have appeared in the literature, and space-filling curves have been applied in diverse areas such as spatial databases, load balancing in parallel computing, improving cache utilization in computations on large matrices, finite element methods, image compression, and combinatorial optimization.

Space-filling curves are usually recursive constructions that map the unit interval $[0,1]$ to a subset of $\Reals^d$ that has measure larger than zero. In the case of Peano's two-dimensional curve, the subset of $\Reals^d$ that is ``filled'' is the unit square. Peano's curve is based on subdividing this unit square into a grid of $3 \times 3$ square cells, and simultaneously subdividing the unit interval into nine subintervals. Each subinterval is then matched to a cell; thus Peano's curve traverses the cells one by one in a particular order. The procedure is applied recursively to each subinterval-cell pair, so that within each cell, the curve makes a similar traversal (see Figure~\ref{fig:peano2d}(a)). The result is a fully-specified mapping from the unit interval to the unit square. A mapping from $\Reals$ to $\Reals^d$ could then be constructed by inverting the recursion, recursively considering the unit interval and the unit square as a subinterval and a cell of a larger interval and a larger square.

It is straightforward to generalize the definition of Peano's two-dimensional curve, based on a subdivision of a square into nine squares, to a three-dimensional curve, based on a subdivision of a cube into 27 cubes, or even a $d$-dimensional curve, based on a subdivision of a hypercube into $3^d$ hypercubes. Peano's curve has been the curve of choice for certain applications. However, in many other applications preference is given to curves based on subdividing squares (or hypercubes) into only $2^d$ squares (or hypercubes). The cell in which a given point $p$ lies can then be determined by inspecting the binary representations of the coordinates of $p$ bit by bit, and no divisions by three need to be computed. 

\begin{figure}
\centering
\hbox to \hsize{\hfill
(a)~\includegraphics[width=0.4\hsize]{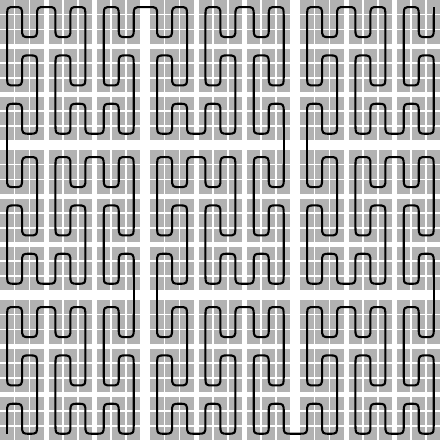}\hfill
(b)~\includegraphics[width=0.4\hsize]{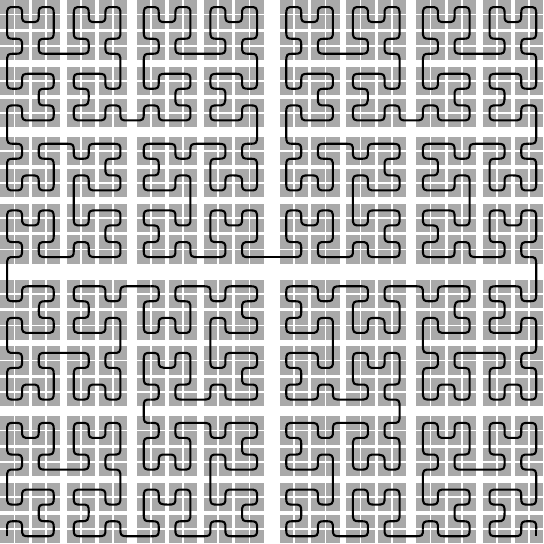}\hfill
}
\caption{(a) A sketch of Peano's space-filling curve.\quad (b) A sketch of Hilbert's space-filling curve.}
\label{fig:peano2d}\label{fig:hilbert2d}
\end{figure}

In response to Peano's publication, Hilbert~\cite{Hilbert} described a two-dimensional space-filling curve based on subdividing a square into four squares (Figure~\ref{fig:hilbert2d}(b)). However, he did not describe how to do something similar in higher dimensions. It is now well-established that it can be done; in fact there are many ways to define a three-dimensional curve based on subdividing a cube into eight octants. Which of these three-dimensional curves would be the three-dimensional ``Hilbert'' curve of choice depends on what properties of a Hilbert curve are deemed essential and what quality measures of the space-filling curve one would like to optimize.

In this paper, we consider a number of properties and quality measures. The quality measures include the $L_\infty$-, $L_2$ and $L_1$-locality measures studied by, among others, Gotsman and Lindenbaum~\cite{Gotsman}, Chochia et al.~\cite{Chochia}, and Niedermeier et al.~\cite{Niedermeier,Niedermeier-Manhattan}; the surface-to-volume measure studied by, for example, Hungersh\"ofer and Wierum~\cite{Hungershoefer}; and the bounding-box quality measures studied by Haverkort and Van Walderveen~\cite{Haverkort}. The one property which we consider to be essential is that the curve is \emph{vertex-continuous}---in effect this means that any pair of consecutive intervals of $[0,1]$ is mapped to a pair of adjacent regions in the unit cube. (Without this property, a curve would be infinitely bad according to most of the locality measures considered.)

The findings presented in this paper include the following:\begin{itemize}
\item Suppose we define a \emph{true} Hilbert space-filling curve to be a space-filling curve that satisfies a certain minimal set of natural requirements such that the only true two-dimensional Hilbert curve is indeed Hilbert's curve. Under this definition there are 10\,694\,807 different three-dimensional true Hilbert curves (modulo rotation, reflection, and reversal). These include a curve that does not start and end in corners of the unit cube and has very good locality properties, and a curve that traverses the points of many facets of the unit cube in the same order as a two-dimensional Hilbert curve would do.
\item If one also allows curves that are described by a system of multiple recursive rules, then the number of different curves is infinite.
    The best previously known upper bounds on the best achievable $L_\infty$-, $L_2$-, and $L_1$-locality measures of three-dimensional generalized Hilbert curves were 29.2, 33.2, and 98.4, respectively. New curves and proofs represented in this paper improve these bounds to 9.45, 18.3, and 75.6, respectively. We prove that no three-dimensional generalized Hilbert curve can have better $L_\infty$-locality.
\end{itemize}

Below, we will first see how exactly we can define Hilbert-like space-filling curves, and then we describe in detail what properties and quality measures of Hilbert-like space-filling curves we consider in this paper. After that we explain how, for some subsets of these properties, we can enumerate all space-filling curves that have these properties, and then we present the most interesting three-dimensional space-filling curves found.

\section{Defining and using space-filling curves}

In general, we can define an order (\emph{scanning order}) of points in $d$-dimensional space as follows. We give a set of rules, each of which specifies (i) how to subdivide a region in $d$-dimensional space into subregions; (ii) what is the order of these subregions; and (iii) for each subregion, which rule is to be applied to establish the order within that subregion. We also specify a unit region, and we indicate what rule is used to subdivide and order it.

\begin{figure}
\centering
\hbox to \hsize{\hfill
(a)\includegraphics[scale=0.7]{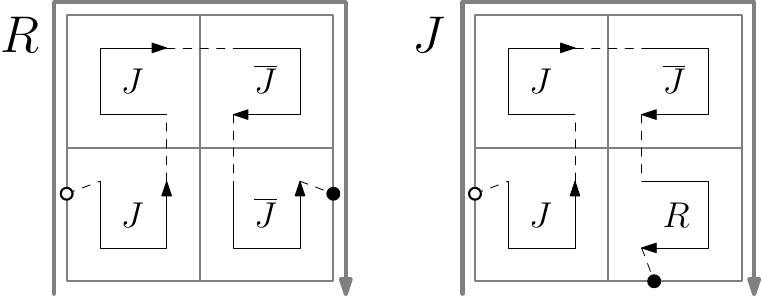}\hfill
(b)\includegraphics[scale=0.7]{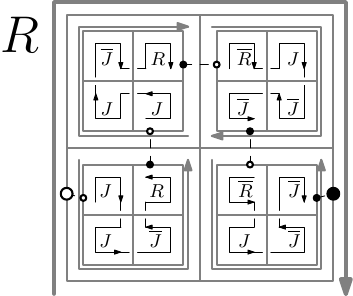}\hfill
}

\vskip\baselineskip
\hbox to \hsize{\hfill
(c)\quad\includegraphics[height=0.35\hsize]{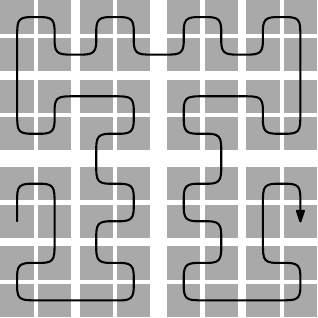}\hfill
(d)\quad\includegraphics[height=0.35\hsize]{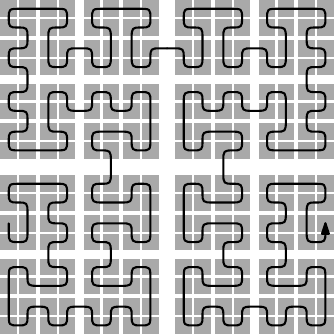}\hfill
}
\caption{%
(a) Definition of a scanning order corresponding to a variant of the $\beta\Omega$-curve.\quad
(b) Applying the rules recursively within \textsf{R}.\quad
(c) The resulting order of the $8 \times 8$-grid.\quad
(d) Order of the $16 \times 16$-grid. On each level, the last cell of each quadrant shares a vertex (and even an edge) with the first cell of the next quadrant.}
\label{fig:betaomega}
\end{figure}

Figure~\ref{fig:betaomega}(a) shows an example (the scanning order depicted corresponds to a section of the $\beta\Omega$-curve by Wierum~\cite{Wierum}). Each rule is identified by a letter, and pictured by showing a unit square, its subdivision into quadrants, the scanning order of the quadrants (by a directed curve along the outer vertices of the quadrants), and the rules applied to the quadrants (by letters). Within each quadrant, there is a scaled copy of the \emph{base pattern}: the curve along the outer vertices of the quadrants. Variations of the rules that consist of simply rotating or mirroring the order within quadrants are indicated by rotating or mirroring the pattern. Variations of rules that consist of reversing the order of cells within a quadrant are indicated by reversing the direction of the pattern and by an overscore above the letter identifying the rule\footnote{It is necessary to do both: reversing the direction of the pattern \emph{and} including the overscore. Figure~\ref{fig:betaomega}(a) illustrates this. Rule $J$ is asymmetric because different rules are applied to the first and the last quadrant. It is therefore important to distinguish between reflected, unreversed applications of $J$, as in the lower left quadrant of rule $R$, and unreflected, reversed applications of $J$, as in the lower right quadrant. Since the base pattern of $J$ is symmetric, it is only the overscore that indicates which of these alternatives is meant.}. Figure~\ref{fig:betaomega}(b) shows how applying these rules recursively, starting with rule \textsf{R}, results in an order of the subquadrants within each quadrant. Combined with the order of the quadrants, this gives an order of all squares in a $8 \times 8$ grid, which is sketched by the curve shown in Figure~\ref{fig:betaomega}(c). By expanding the recursion further, one may order the cells of an arbitrarily fine grid. The polygonal curve that connects the centres of these cells in order thus forms an arbitrarily fine approximation of a space-filling curve; thus the definition of the scanning order is also a definition of a space-filling curve.

When the scanning order is refined to an infinitely fine grid, the first and the last cell visited shrink to points on the boundary of the unit square. We will call these points the \emph{entrance gate} and the \emph{exit gate} of the curve; in Figure~\ref{fig:betaomega} their location is indicated by the white dots and the black dots, respectively. Note that these dots are only shown for clarity, they are a consequence of the definition of the scanning order, rather than a part of the definition.
The letters indicating the rules may be omitted if there are no reversals and no two rules have the same base pattern modulo rotation and reflection.

Figures \ref{fig:decomposable-curves} to~\ref{fig:jupiter} show examples of definitions of three-dimensional space-filling curves.

Note that for many practical applications of space-filling curves it is not necessary to actually draw a curve. It is often enough to be able to decide for any two given points $p$ and $q$, which of these appears first on the curve---or more precisely, which point comes first in the scanning order. This can be decided by expanding the recursion only to the smallest depth at which $p$ and $q$ lie in different quadrants or octants.\footnote{A technical problem with this is that, down from some depth of recursion, a given point $p$ may always lie on the boundary of two or more quadrants. This may create ambiguity about which of two points $p$ and $q$ comes first. This ambiguity can easily be resolved by using a consistent tie-breaking rule, for example, always assign $p$ to the quadrant to its upper left, or to the quadrant that comes first in the order.}

\section{Properties of generalized Hilbert curves}
\label{sec:properties}

The space-filling curve curves studied in this paper all have the following properties:\begin{itemize}
\item They are defined by a rule system as described above, in which the unit region is always the $d$-dimensional unit hypercube, and all rules subdivide hypercubes into $2^d$ smaller hypercubes of equal size. When $d=2$, the hypercubes are squares and the subcubes are called \emph{quadrants}; when $d=3$, the hypercubes are cubes and the subcubes are called \emph{octants}.
\item The scanning order is \emph{vertex-continuous}: this means that if the scanning order visits a set of hypercubes $C_1,...,C_k$, in that order, and we refine the scanning order in each of these hypercubes recursively to depth $m$ so that we obtain an order on $k \cdot 2^{md}$ cells, for any $m \geq 1$, then the last cell within $C_{i-1}$ shares at least one vertex with the first cell within $C_i$, for all $1 < i \leq k$ (see, for example,  Figures \ref{fig:peano2d} and~\ref{fig:betaomega}(d)).
\end{itemize}

The reason to require the second property will be explained at the end of the next section. In this paper, a curve that has the above properties will be called a \emph{generalized Hilbert curve}. More specifically, such a curve will be called a \emph{mono-Hilbert curve} (or \emph{mH-curve}) if the defining rule system contains only one rule, and a \emph{poly-Hilbert curve} (or \emph{pH-curve}) if the defining rule system consists of more than one rule. With these definitions, in one dimension, there is no other generalized Hilbert curve than the trivial curve that traverses the unit line segment from one end to the other. In two dimensions, there are no mono-Hilbert curves other than Hilbert's original curve.

Besides the properties described above, there are a number of other properties we may be interested in, or which may be convenient for case distinctions when exploring the realm of mono-Hilbert curves and poly-Hilbert curves:\begin{itemize}
\item Mono-Hilbert curves may be \emph{order-preserving}: this means that the curve can be described by a defining rule that is not reversed in any octant.\footnote{For poly-Hilbert curves we do not consider this property, since reversals can always be eliminated by using additional rules that describe the reversed versions of the original rules.}
\item Curves may be \emph{symmetric}: this means that they are the same as their own reversed version that is reflected in a line (in 2D) or plane (in 3D) or rotated 180 degrees around a point (in 2D) or line (in 3D). Note that we do not require this line or plane to be parallel to any edges or faces of the unit square/cube: it could also be diagonal. Symmetric mono-Hilbert curves are always order-preserving.
\item $d$-Dimensional mono-Hilbert curves may be \emph{binary-decomposable}: this means that the scanning order can be defined by $d$ rules $\mathsf{R}_0, ..., \mathsf{R}_{d-1}$, where each rule $R_i$ divides a region by an axis-parallel cutting plane into \emph{two} regions of equal size, to each of which rule $R_{(i+1)\bmod d}$ is applied. (We will not require the rules $\mathsf{R}_0, ..., \mathsf{R}_{d-1}$ to be order-preserving, even if the curve itself is order-preserving.) Figure~\ref{fig:decomposable} shows a binary decomposition of the two-dimensional Hilbert curve.
\begin{figure}
\centering
\hbox to \hsize{\hfill
(a)\quad\includegraphics[scale=0.6]{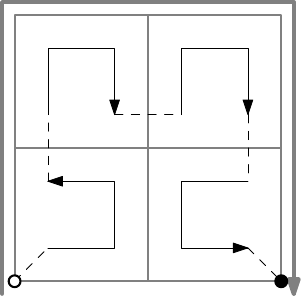}\hfill
(b)\includegraphics[scale=0.6]{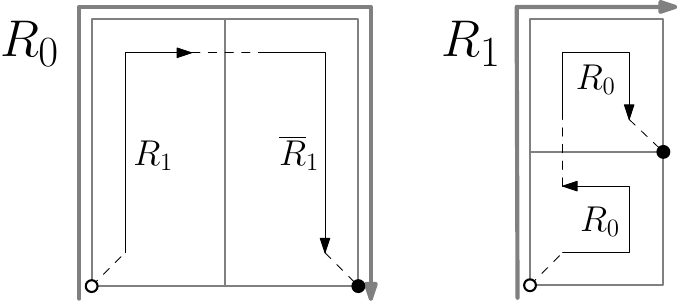}\hfill
}
\caption{(a) Definition of the two-dimensional Hilbert curve. Because there is only one rule and the definition is order-preserving, the letters identifying the rule are omitted. (b) Binary decomposition of the two-dimensional Hilbert curve.}
\label{fig:decomposable}\label{fig:vertexgated}
\end{figure}
\item $d$-Dimensional generalized Hilbert curves may be \emph{downward-compatible} (or \emph{partially down-ward-compatible}): this means that for each (or some) of the $d$ coordinate axes the $d$-dimensional unit cube has a $(d-1)$-dimensional face $F$ orthogonal to that axis, so that the $d$-dimensional curve visits the points of $F$ in the same order as the order in which they would be visited by a (possibly rotated, reflected or reversed) $(d-1)$-dimensional downward-compatible Hilbert curve filling~$F$; we define the trivial one-dimensional generalized Hilbert curve (the traversal of a line segment from one end to the other) to be downward-compatible by definition.
    Note that downward-compatibility cannot be verified by only looking at the order in which the cells of a $4^d$ grid are traversed and this may sometimes be misleading: one needs to make sure that the $(d-1)$-dimensional Hilbert order on the face(s) is maintained also when the grid is refined recursively.
\begin{figure}
\centering
\hbox to \hsize{\hfill
\includegraphics[scale=0.8]{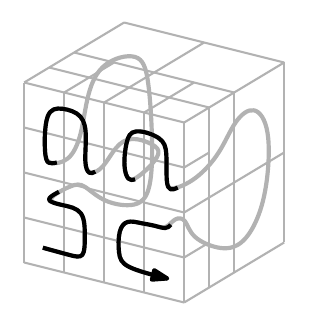}\hfill
}
\caption{A sketch of a partially downward-compatible three-dimensional generalized Hilbert curve.}
\label{fig:compatible}
\end{figure}
\item $d$-Dimensional curves may be \emph{$k$-continuous}, for $0 \leq k < d$: this means that if the scanning order visits a set of cells $C_0,...,C_k$, in that order, and we refine the scanning order in each cell recursively so that we obtain an order on $k \cdot 2^{md}$ subcells, for any $m \geq 1$, then the last subcell within $C_{i-1}$ shares at least one $k$-dimensional face with the first subcell within $C_i$, for all $1 \leq i \leq k$. We may also describe curves that are $(d-1)$-continuous as \emph{face-continuous}, and we may describe curves that are 1-continuous as \emph{edge-continuous}.
\item Curves may be \emph{vertex-gated}: this means that in each rule, the gates are vertices of the unit cube. Similarly, curves may be \emph{edge-gated} (the gates lie in the interior of edges of the unit cube), \emph{face-gated} (the gates lie in the interior of faces of the unit cube), or even, for example, \emph{vertex-edge-gated} (the gate of entry lies at a vertex and the gate of exit lies in the interior of an edge of the unit cube, or the other way around). For example, the two-dimensional Hilbert curve is vertex-gated (Figure~\ref{fig:vertexgated}(a)), whereas the $\beta\Omega$-curve is edge-gated (Figure \ref{fig:betaomega}).
\item Mono-Hilbert curves may be \emph{non-reflecting}: this means that the curve can be described by a defining rule that is possibly rotated and/or reversed, but never reflected in any octant.
\end{itemize}

Hilbert's two-dimensional curve (Figure~\ref{fig:decomposable}(a)) is order-preserving, symmetric, binary-decom-posable, downward-compatible, face-continuous, and vertex-gated\footnote{Note that within the first (bottom left) and last (bottom right) quadrant we must make a left turn, while the base pattern makes a right turn. A left turn can be obtained either by reflecting the base pattern, or by reversing the base pattern. Thus we can choose to say that Hilbert's two-dimensional curve is order-preserving, or that Hilbert's two-dimensional curve is non-reflecting---but it cannot be both.}. We may now ask what properties a three-dimensional space-filling curve should have to be called a \emph{true} three-dimen-sional Hilbert curve. Ultimately, the answer is a matter of taste, and I can only give some considerations that one might want to take into account when answering this question. Requiring a three-dimensional Hilbert curve to have generalizations of \emph{all} properties of the two-dimensional Hilbert curve seems too restrictive: one can always find a property that holds for the two-dimensional curve but cannot be generalized. In fact, in three dimensions, no downward-compatible, binary-decomposable mono-Hilbert curve exists. Thus, a more reasonable choice of requirements for a \emph{true} Hilbert curve may be a minimal set of requirements such that the two-dimensional Hilbert curve exists and is unique (modulo rotation, reflection and reversal). The requirements that define a mono-Hilbert curve (one rule, subdividing a square into four squares, with vertex-continuity) are such a choice: if we drop any of these requirements, the two-dimensional Hilbert curve would not be unique anymore. Regardless of what constitutes, according to our taste, a true three-dimensional Hilbert curve, we may, of course, still prefer to work with other three-dimensional curves---for example, poly-Hilbert curves---in particular applications.

\section{Quality measures of space-filling curves}
\label{sec:measures}

In addition to the properties mentioned above, we consider a number of curve quality measures. To define these quality measures we need the following definitions: given two points $p$ and $q$ in the unit cube, let $\delta_i(p,q)$ be the $L_i$-distance between $p$ and $q$, let $C(p,q)$ be the set of points that appear on the curve between $p$ and $q$, and let $\mathrm{vol}(S)$, $\mathrm{diam}_i(S)$, $\bbox(S)$, $\mathrm{bball}_i(S)$, and $\mathrm{surface}(S)$ be the volume, $L_i$-diameter, the minimum axis-parallel bounding box, the minimum bounding $L_i$-ball, and the measure of the boundary of the set $S$, respectively. We can now define the following quality measures of a $d$-dimensional space-filling curve:\begin{itemize}
\item \emph{$L_i$-locality} or $\mathrm{WL}_i$ (for $i \in \{1,2,\infty\}$): the maximum, over all pairs of points $p,q$ in the unit cube, of $\delta_i(p,q)^d$ $/$ $\mathrm{vol}(C(p,q))$;
\item \emph{$L_i$-diameter ratio} or $\mathrm{WD}_i$ (for $i \in \{1,2,\infty\}$): the maximum, over all pairs of points $p,q$ in the unit cube, of $\mathrm{diam}_i(C(p,q))^d$ $/$ $\mathrm{vol}(C(p,q))$;
\item \emph{$L_i$-bounding ball ratio} or $\mathrm{WBB}_i$ (for $i \in \{1,2,\infty\}$): the maximum, over all pairs of points $p,q$ in the unit cube, of $\mathrm{vol}(\mathrm{bball}_i(C(p,q)))$ $/$ $\mathrm{vol}(C(p,q))$;
\item \emph{surface ratio} or \WS: the maximum of $(\mathrm{surface}(C(p,q))/2d)^{d/(d-1)}$ $/$ $\mathrm{vol}(C(p,q))$;
\item \emph{bounding-box volume ratio} or \WBV: the maximum of $\mathrm{vol}(\mathrm{bbox}(C(p,q)))$ $/$ $\mathrm{vol}(C(p,q))$;
\item \emph{bounding-box surface ratio} or \WBS: the maximum of $(\mathrm{surface}(\mathrm{bbox}(C(p,q)))/2d)^{d/(d-1)}$ $/$ $\mathrm{vol}(C(p,q))$.
\end{itemize}
In fact, the $L_i$-locality and the $L_i$-diameter ratio of a space-filling curve are always equal\footnote
{Consider two points $p,q$, and let $r,s$ be two points of $C(p,q)$ that determine the diameter of $C(p,q)$ and are as close to each other as possible along the curve. Since $C(r,s) \subseteq C(p,q)$ and $\delta_i(r,s) \geq \delta_i(p,q)$, reducing the curve section under consideration from $C(p,q)$ to $C(r,s)$ can only increase its locality ratio $\delta_i(p,q)^d / \mathrm{vol}(C(p,q))$ and its diameter ratio $\mathrm{diam}_i(C(p,q))^d / \mathrm{vol}(C(p,q))$; when $C(p,q)$ shrinks to $C(r,s)$ both will rise to the same value $\delta_i(r,s)^d / \mathrm{vol}(C(r,s))$. Therefore, $\mathrm{WL}_i$ and $\mathrm{WD}_i$ are determined by the same point pairs and have the same value.}.
For $i = \infty$ and $i = 1$, the $L_i$-diameter ratio and the $L_i$-bounding ball ratio are always equal as well, modulo a fixed constant factor\footnote
{This is shown by similar arguments as for the equality of $L_i$-locality and $L_i$-diameter ratio, using the observation that for $i = \infty$ and $i = 1$, the diameter of the minimum bounding $L_i$-ball of any set $S$ is equal to the $L_i$-diameter of $S$.}.
I conjecture that the same holds for $i = 2$: we can prove this for two-dimensional space-filling curves\footnote
{Consider two points $p$ and $q$, and let $B$ be a minimal set of points of $C(p,q)$ such that $\mathrm{bball}_2(B) = \mathrm{bball}_2(C(p,q))$. Suppose $B$ consists of three points $r$, $s$, and $t$, in the order in which they occur along the space-filling curve between $p$ and $q$. Let $\rho$, $\sigma$ and $\tau$ be the angles of the triangle $\triangle rst$ at $r$, $s$, and $t$, respectively. Note that $\rho, \sigma, \tau < \pi/2$, and the diameter of $\mathrm{bball}_2(C(p,q))$ is $\delta_2(r,s)/\sin\tau = \delta_2(s,t)/\sin\rho$.

Now, suppose for the sake of contradiction, that neither $C(r,s)$ nor $C(s,t)$ has a $L_2$-bounding ball ratio that is worse than that of $C(p,q)$. Thus
$\mathrm{vol}(\mathrm{bball}_2(C(r,s)))/\mathrm{vol}(C(r,s)) \leq \mathrm{vol}(\mathrm{bball}_2(C(p,q)))/\mathrm{vol}(C(p,q))$, which implies
$\delta_2(r,s)^2/\mathrm{vol}(C(r,s)) \leq  (\delta_2(r,s)^2/\sin^2\tau)/\mathrm{vol}(C(p,q))$, and thus
$\sin^2\tau \leq \mathrm{vol}(C(r,s))/\mathrm{vol}(C(p,q))$; similarly:
$\sin^2\rho \leq \mathrm{vol}(C(s,t))/\mathrm{vol}(C(p,q))$; adding these up yields:
$\sin^2\tau + \sin^2\rho \leq 1$. However, since all angles of $\triangle rst$ are acute angles, we have $\rho = \pi - \sigma - \tau > \pi/2 - \tau$ and $\sin^2\tau + \sin^2\rho > \sin^2\tau + \sin^2 (\pi/2 - \tau) = 1$. Therefore $p$ and $q$ cannot determine the $L_2$-bounding ball ratio.

Furthermore, suppose that neither $C(r,s)$ nor $C(s,t)$ has a $L_2$-diameter ratio that is worse than that of $C(p,q)$. Thus
$\mathrm{diam}_2(C(r,s))^2/\mathrm{vol}(C(r,s)) \leq \mathrm{diam}_2(C(p,q))^2/\mathrm{vol}(C(p,q))$, which again implies $\delta_2(r,s)^2/\mathrm{vol}(C(r,s)) \leq  (\delta_2(r,s)^2/\sin^2\tau)/\mathrm{vol}(C(p,q))$; proceeding in the same way as for the $L_2$-bounding ball ratio, we find that $p$ and $q$ cannot determine the $L_2$-diameter ratio either.

Hence both the $L_2$-bounding ball ratio and the $L_2$-diameter ratio must be determined by a pair of points $p,q$ for which $B$ consists of only two points $r$ and $s$, in which case the $L_2$-diameter of $C(p,q)$ is also the diameter of $\mathrm{bball}_2(C(p,q))$.},
but I have not found a proof for three-dimensional space-filling curves yet.

In a previous publication on two-dimensional space-filling curves we described algorithms to compute bounds on $\mathrm{WL}_i$, \WBV, and \WBS for any given curve~\cite{Haverkort}. We have now implemented higher-dimensional versions of these algorithms, including an algorithm to compute \WS, and used these algorithms to analyse the curves discussed in the next sections of this paper.

Note that a curve that is not vertex-continuous would, at some level of recursion, visit two cells $C_1$ and $C_2$ that are consecutive in the order but not adjacent in space. One could say the curve \emph{jumps} from $C_1$ to $C_2$. It is generally assumed that jumps cause bad performance in many applications. This also shows in the quality measures described above. To see this, consider refining the recursion, and consider the last subcell $C'_1$ in $C_1$ and the first subcell $C'_2$ in $C_2$. These, too, are consecutive in the order, but not adjacent in space: the distance $\delta_i(p,q)$ between any pair of points $p \in C'_1, q \in C'_2$ is at least the distance between $C_1$ and $C_2$, while the size of the set of points $C(p,q)$ between $p$ and $q$ along the curve is at most $\mathrm{vol}(C'_1) + \mathrm{vol}(C'_2)$. By refining the recursion, getting smaller cells $C'_1$ and $C'_2$, we can make the ratio $\delta_i(p,q)^d / \mathrm{vol}(C(p,q))$ arbitrarily big. Hence, a curve that is not vertex-continuous would have an infinitely bad score on the $L_i$-locality measures mentioned above. Similarly, it would have an infinitely bad score on the bounding box quality measures. This is why we restrict our attention to vertex-continuous curves.

\section{Previous results on the quality of three-dimensional generalized Hilbert curves}

Many authors have studied locality measures for two-dimensional space-filling curves (for an overview, see Haverkort and Van Walderveen~\cite{Haverkort}), but for three-dimensional curves results are few. Gotsman and Lindenbaum~\cite{Gotsman} proved that a three-dimensional Hilbert curve has \WLEuc at most $48\sqrt 6 \approx 117.6$, and reported to have found that \WLEuc is at most 23 by computer simulation. From their publication it is not clear which curve or which curves exactly they evaluated, but their results are confirmed by our computations: the curve which we will identify as A26.0010\,1011.1011\,0011 (see Figure~\ref{fig:A26.2b.b3}) fits the scope of their publication and we found that its \WLEuc-value is 22.9.

Niedermeier et al.~\cite{Niedermeier} proved lower bounds of $\WLMax \geq 8.25$, $\WLEuc \geq 11.1$ and $\WLMan \geq 42.625$ for space-filling curves that fill a cube, and conjectured that stronger lower bounds of $\WLMax \geq 9$, $\WLEuc \geq 12.39$ and $\WLMan \geq 54$ would hold---they could prove the stronger lower bounds only for curves that fill a cube such that the exit gate coincides with the entrance gate. They also analysed the three symmetric, non-reflecting, face-continuous three-dimensional order-preserving mono-Hilbert curves; in the numbering scheme described in Section~\ref{sec:OmHnumbering} below, these curves will be identified as A16.00.00, A18.00.00, and A26.00.00. In particular, for A26.00.00 (Figure~\ref{fig:A26-curves}(b)) they found $\WLMax \leq 29.2$, $\WLEuc \leq 33.2$, and $\WLMan \leq 98.4$; we now obtain the following more precise bounds for this curve: $\WLMax = 24.22$; $\WLEuc = 26.23$; and $\WLMan = 98.34$. For A16.00.00 and A18.00.00 the numbers are higher.

Chochia and Cole~\cite{Chochia} prove that all face-continuous three-dimensional order-preserving mono-Hilbert curves have $\WLMan$ at least $85\frac13$ and at most $128$. They present a new face-continuous three-dimensional poly-Hilbert curve, the $H^*$-curve, which they compared to the non-reflecting, face-continuous three-dimensional order-preserving mono-Hilbert curves by estimating  their \WLMan-values for pairs of points in a $32\times 32\times 32$ grid. Chochia and Cole describe their curve by giving four rules in which the first and the last octant are left open, together with a table with 24 entries that specifies how to choose one out of three possibilities for each end, depending on the context. In effect they thus describe $4 \cdot 3 \cdot 3 = 36$ combinations; however, expanding the recursion one would find that only ten of these are actually used; thus the $H^*$-curve is a poly-Hilbert curve with ten rules. With our algorithm we found that it has $\WLMan = 84.0$, which is indeed better than any mono-Hilbert curve.

\section{Generating three-dimensional mono-Hilbert curves}
\label{sec:generating}

There are only a finite number of different three-dimensional mono-Hilbert curves, so one may try to generate all of them.

The recursive definitions of space-filling curves as shown in Figures \ref{fig:decomposable-curves} to~\ref{fig:jupiter} focus on describing the order in which octants within cubes are visited. The number of different ways to order eight octants is 840, so we could consider generating all of them, and then generating all possibilities to rotate, reflect and/or reverse the order within each octant such that a vertex-continuous curve results.%
\footnote{We can identify the octants by coordinate triples from $\{0,1\}^3$. Modulo reflections and rotations, we may assume that the traversal starts with $(0,0,0)$, that $(0,0,1)$ appears before $(0,1,0)$, and that $(0,1,0)$ appears before $(1,0,0)$. The remaining four octants could be ordered in $4!$ ways, and they could be mixed in with the sequence $(0,0,1)$, $(0,1,0)$, $(1,0,0)$ in $7!/3!4!$ ways. Thus there are $7!/3! = 840$ different orders.}
However, if we would take this approach, it would be quite complicated to figure out if, in recursion, the exit gate of each octant matches the entrance gate of the next octant.

Therefore we take an alternative approach, generalizing the approach taken by Alber and Niedermeier~\cite{Alber}. Instead of first generating only the order in which octants are visited, we first fix the entrance and exit gates of the unit cube, and then generate possible sequences of, alternatingly, octants and gates between them. While doing so, we make sure that every triple of an octant with its entrance gate and its exit gate can be obtained by a transformation of the unit cube with its entrance and exit gates. Thus we find a set of possible \emph{connection schemes}. A connection scheme is not a space-filling curve definition yet: this is because for each octant~$C$, there may be multiple transformations that map the unit cube and its gates to octant~$C$ and its gates; for a full definition of a space-filling curve we have to choose, for each octant, which of these transformation to use. Thus a connection scheme may give rise to a number of different actual space-filling curves. Figure~\ref{fig:scheme-example} shows such a connection scheme, with two examples of space-filling curve definitions that match the scheme.

\begin{figure}
\centering
\hbox to \hsize{\hfill
(a)\includegraphics[scale=0.8]{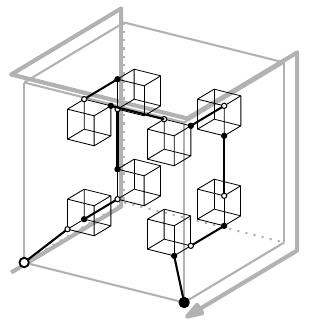}\hfill
(b)\includegraphics[scale=0.8]{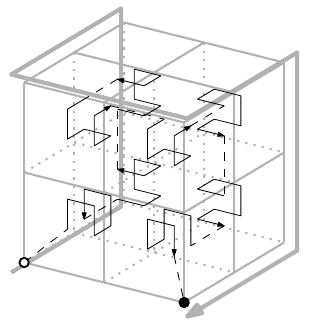}\hfill
(c)\includegraphics[scale=0.8]{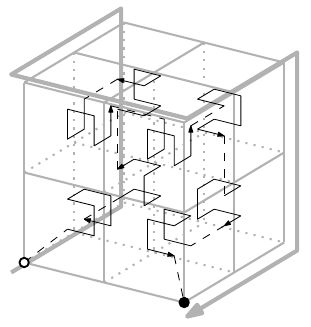}\hfill
}
\caption{(a) An example of a connection scheme (A26). (b)  A curve (A26.00.0000\,0000) that matches the connection scheme: it is a symmetric, non-reflecting curve with optimal \WBV among order-preserving mono-Hilbert curves. (c) Another curve (A26.00.1101\,1011) that matches the connection scheme: it is a downward-compatible mono-Hilbert curve that shows two-dimensional Hilbert orders on all faces except the back face of the cube.}
\label{fig:scheme-example}\label{fig:A26-curves}
\end{figure}

Of course the viability of this approach depends on the ease with which we can generate all realizable combinations of entrance and exit gates for the unit cube. Fortunately the number of possible combinations is fairly limited. We will adopt a coordinate system in which the lower left front vertex of the unit cube has coordinates $(0,0,0)$, the lower right front vertex is at $(1,0,0)$, the lower left back vertex is at $(0,1,0)$, and the upper left front vertex is at $(0,0,1)$. We can prove that for order-preserving mono-Hilbert curves, we do not need to consider any other combinations than the following:\begin{itemize}
\item For gates at vertices, we may assume that the entrance gate is at $(0,0,0)$ and the exit gate is either at $(1,0,0)$ or at $(0,1,1)$ (proof in the appendix, Lemma~\ref{lem:vertex-vertex-gates}).
\item For gates in the interior of edges, we may assume that the entrance gate is at $(\frac13,0,0)$ and the exit gate is at $(1,\frac13,1)$ (proof in the appendix, Lemma~\ref{lem:edge-edge-gates});
\item For gates in the interior of faces, we may assume that the entrance gate is at $(0,\frac13,\frac13)$ and the exit gate is at $(\frac23,\frac13,0)$  (proof in the appendix, Lemma~\ref{lem:face-face-gates});
\end{itemize}
Obviously, in order-preserving mono-Hilbert curves, both gates must be of the same type (at a vertex, in the interior of an edge, or in the interior of a face). For non-order-preserving mono-Hilbert curves this is not necessarily the case and we need to consider the following combinations:\begin{itemize}
\item For one gate at a vertex and one gate in the interior of an edge: entrance at $(0,0,0)$, exit at $(1,\frac12,0)$ (proof in the appendix, Lemma~\ref{lem:vertex-edge-gates});
\item For one gate at a vertex and one gate in the interior of a face: entrance at $(0,0,0)$, exit at $(1,\frac12,\frac12)$ (proof in the appendix, Lemma~\ref{lem:vertex-face-gates});
\item For one gate in the interior of an edge and one gate in the interior of a face: no combination possible (proof in the appendix, Lemma~\ref{lem:edge-face-gates});
\end{itemize}

Below we will discuss the number of curves of each type. In total, we find that there are 10\,694\,807 three-dimensional mono-Hilbert curves.

\paragraph{Type A and B: vertex-gated curves}
There are two types of vertex-gated curves: with the exit gate at $(1,0,0)$ (type A) and with the exit gate at $(0,1,1)$ (type B). Both types have the entrance gate at $(0,0,0)$.

Connection schemes of type A and B can easily be enumerated by exhaustive search: there are 29 connection schemes of type A and 149 connection schemes of type B (Figure~\ref{fig:connection-schemes} shows some of them).

\begin{figure}
\centering
\hbox to\hsize{\includegraphics[width=\hsize]{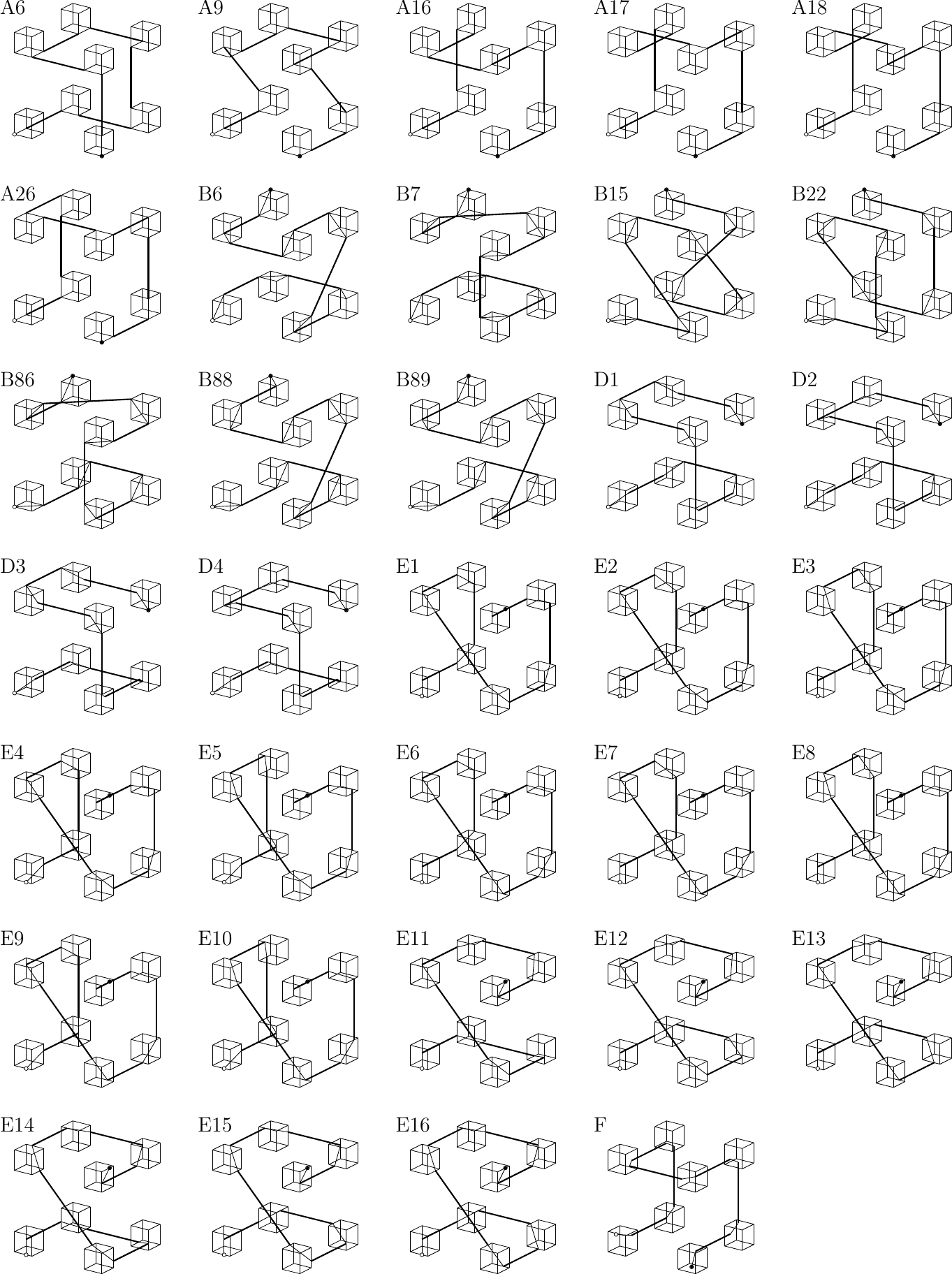}}
\caption{Connection schemes for mono-Hilbert curves: a selection of schemes of type A and B mentioned in this paper, and all schemes of type D, E, and F. Figures of the complete set of schemes of type A and B can be requested from the author.}
\label{fig:connection-schemes}
\end{figure}

With vertex-gated curves, one can tell from the connection scheme alone whether the curve is face-continuous, edge-continuous, or only vertex-continuous. The 29 connection schemes of type A include include 5 schemes for face-continuous curves (3 symmetric and 2 asymmetric)\footnote{The five schemes for face-continuous curves (A6, A16, A17, A18, A26 in Figure~\ref{fig:connection-schemes}) were also described by Alber and Niedermeier. They also described a sixth scheme, version (b) of $\mathrm{Hil}^3_1.\mathrm{B}$ in their paper; however this is the same as version (a) of $\mathrm{Hil}^3_1.\mathrm{B}$ (A6) rotated 180 degrees around a line through the midpoints of the lower front and the upper back edge.}, 21 schemes for edge-continuous curves (7 symmetric and 14 asymmetric), and 3 schemes for vertex-continuous curves (all asymmetric); the 149 connection schemes of type B (18 symmetric and 131 asymmetric) are all only edge-continuous.

We will now analyse the number of different curves (modulo rotation, reflection and reversal) that we can make from these connection schemes. For each octant in any given scheme, we may choose whether to use a forward copy of the curve, or a reversed-and-reflected copy (where the reflection is in the plane containing the gates and the centrepoint of the octant). Given this choice, the location of the gates determines how the curve should be rotated within the octant. Then we still have the choice whether or not to reflect it (that is, whether or not to apply reflection to a forward copy, or to undo the reflection of a reversed copy). In total, there are two choices to make for each octant, which makes 16 choices for the full curve. Thus we get:\begin{itemize}
\item asymmetric curves from asymmetric schemes: $2^{16} = 65\,536$ curves from each scheme;
\item symmetric curves from symmetric schemes: symmetric curves are invariant under reversal-and-reflection, which renders eight choices meaningless; moreover, the reflections of the first four octants should match those of the last four octants, so that there are only four choices left to make; thus we get $2^4 = 16$ symmetric curves from each scheme;
\item asymmetric curves from symmetric schemes: in this case, reversal-and-reflection does have an effect: although it does not change the order in which the suboctants of any octant are traversed (because the connection scheme is symmetric), it reverses the sequence of decisions whether or not to apply reflection when we traverse the octants in order. This sequence should be asymmetrical, otherwise we end up with a symmetric curve after all. Thus we get $2^8 = 256$ combinations of choices which octants to reverse-and-reflect, times $256-16 = 240$ asymmetric combinations of choices which octants to reflect. However, since the connection scheme is symmetric, every curve is generated twice in this way: for every curve we also generate a reversed-and-reflected copy. Hence, in total, we get $256 \cdot 240 / 2 = 30\,720$ asymmetric curves from each asymmetric connection scheme.
\end{itemize}
\label{sec:OmHnumbering}
In total, the number of curves of type A is $(3 + 7) (16 + 30\,720) + (2 + 14 + 3)(65\,536) = 1\,552\,544$ and the total number of curves of type B is $18 \cdot (16 + 30\,720) + 131 \cdot 65\,536 = 9\,138\,464$.
We will identify these curves by names that consist of:\begin{itemize}
\item their type (A or B);
\item their base-10 connection scheme number (see Figure~\ref{fig:connection-schemes});
\item a dash;
\item an eight-digit binary number indicating which octants are reversed (this number has a one for the i-th digit from the right if and only if the $i$-th octant is reversed);
\item an eight-digit binary number indicating which octants are reflected (this number has a one for the i-th digit from the right if and only if the $i$-th octant is reflected).
\end{itemize}
Alternatively, the binary numbers may be given in hexadecimal notation. Note that a reverse-and-reflect operation as described above affects \emph{two} digits in the curve identifier.

From the type-A curves, 223\,280 are face-continuous; 920 are face-continuous and order-preserving\footnote{These are the 1\,536 curves from Alber and Niedermeier~\cite{Alber} with symmetric copies removed.}.

\paragraph{Type C: vertex-edge-gated curves}

An exhaustive search brought up 2\,758 connection schemes with the entrance gate at vertex $(0,0,0)$ and the exit gate half-way an edge at $(1,\frac12,0)$. The location of the gates leaves no freedom with respect to reversals and reflections, so each scheme generates exactly one curve. None of these curves is face-continuous and none of these curves is order-preserving.

\paragraph{Type D: vertex-face-gated curves}

An exhaustive search brought up 4 connection schemes with the entrance gate at vertex $(0,0,0)$ and the exit gate in the middle of a face at $(1,\frac12,\frac12)$, see Figure~\ref{fig:connection-schemes}.
The location of the gates leaves freedom with respect to reflections, but not with respect to reversals, so each scheme generates exactly $2^8 = 256$ curves. Thus there are 1024 curves of type D in total. None of them is face-, or even edge-continuous, and none of them is order-preserving.

\paragraph{Type E: edge-gated curves}

An exhaustive search brought up 16 connection schemes with the entrance gate at $(\frac13,0,0)$ and the exit gate at $(1,\frac13,1)$, see Figure~\ref{fig:connection-schemes}. The location of the gates leaves no freedom with respect to reversals and reflections, so each scheme generates exactly one curve. None of these curves is face-continuous and none of them is order-preserving.

\paragraph{Type F: face-gated curves}

An exhaustive search brought up exactly one connection scheme with the entrance gate at $(0,\frac13,\frac13)$ and the exit gate at $(\frac23,\frac13,0)$, see Figure~\ref{fig:connection-schemes}. The location of the gates leaves no freedom with respect to reversals and reflections, so there is, in fact, only one curve of type F. Since it is face-gated, it is also face-continuous. It is not order-preserving.

\section{Noteworthy three-dimensional mono-Hilbert curves}

\begin{figure}
\centering
\hbox to \hsize{\hfill
\includegraphics[scale=0.8]{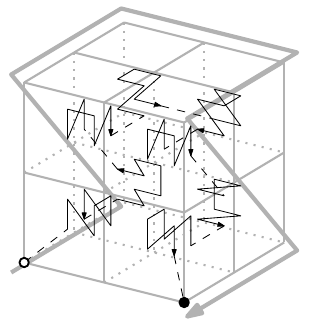}\hfill
\includegraphics[scale=0.8]{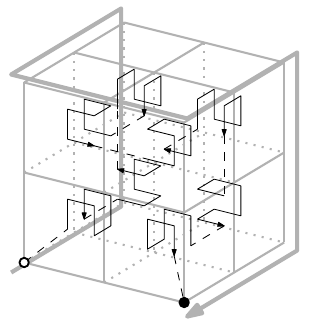}\hfill
}

\addvspace\baselineskip
\hbox to \hsize{\hfill
\includegraphics[scale=0.8]{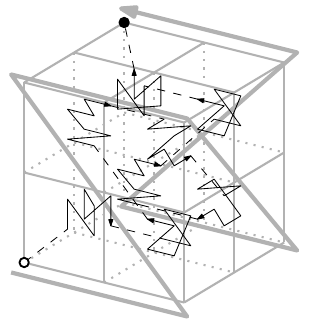}\hfill
\includegraphics[scale=0.8]{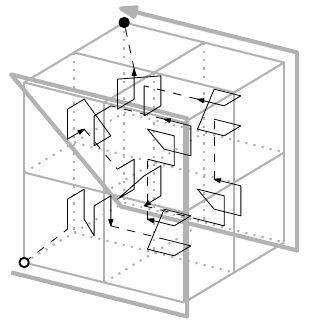}\hfill
}
\caption{The three-dimensional binary-decomposable generalized Hilbert curves. From each of these curves variations can be made by reflecting all of the 1st, 4th, 5th and the 8th octants, and/or all of the 2nd, 3rd, 6th and 7th octants.}
\label{fig:decomposable-curves}
\end{figure}

\paragraph{Binary-decomposable curves}\label{sec:binary-decomposable}
There are 16 binary-decomposable mono-Hilbert curves. Four of these curves are A9.00.0000\,0000, A16.00.0011\,1100, B15.00.1100\,0011, and B22.00.0000\,0000, shown in Figure~\ref{fig:decomposable-curves}. From each of them three variations can be made by reflecting all of the 1st, 4th, 5th and the 8th octants, and/or all of the 2nd, 3rd, 6th and 7th octants. Thus there are 16 binary-decomposable order-preserving mono-Hilbert curves in total. 
All binary-decomposable curves are symmetric, and therefore, order-preserving.

\paragraph{Downward-compatible curves}
There are ten downward-compatible mono-Hilbert curves; all are based on connection scheme A26.
Curve A26.00.1101\,1011 (Figure~\ref{fig:A26-curves}(c)) is an order-preserving mono-Hilbert curve with two-dimensional Hilbert orders on the left, right, top, bottom and front face of the cube.
A26.00.1001\,1001 is similar but sacrifices the bottom face for better locality measures on most aspects.
The other eight downward-compatible mono-Hilbert curves are not order-preserving, and show two-dimensional Hilbert orders on only three faces (right, front, and top).



\bigskip
Hilbert curves are often associated with face-continuous, vertex-gated, and maybe even order-preserving mono-Hilbert curves: above we saw that there are 920 such curves in three dimensions. However, we argued that there are more than 10\,000 as many other mono-Hilbert curves that could also be considered to be three-dimensional generalizations of Hilbert's curve, see Table~\ref{tab:extracurves}.
What do these extra curves bring us? In the following paragraphs, we will refer to tables with locality and bounding-box quality measures. The bounds are rounded numbers; the last digits may be off by one.

\begin{table}
\centering
\begin{tabular}{r|ccc}
number of curves & face-continuous & vertex-gated & order-preserving \\
\hline
             920 &   $\bullet$     &   $\bullet$  &    $\bullet$     \\
               1 &   $\bullet$     &              &                  \\
         41\,288 &                 &   $\bullet$  &    $\bullet$     \\
\hline
        222\,360 &   $\bullet$     &   $\bullet$  &                  \\
    10\,426\,440 &                 &   $\bullet$  &                  \\
          3\,798 &                 &              &                  \\
\hline
               0 &   $\bullet$     &              &    $\bullet$     \\
               0 &                 &              &    $\bullet$     \\
\end{tabular}
\caption{Numbers of mono-Hilbert curves with different combinations of properties.}
\label{tab:extracurves}
\end{table}

\paragraph{Non-face-continuous curves}
From an implementation point of view, order-preserving curves are arguably simpler than non-order-preserving curves. Within the realm of order-preserving curves, there are 920 face-continuous curves and 41\,288 non-face-continuous curves---all order-preserving curves are vertex-gated. Although conventional wisdom holds that face-continuity is good for locality, the flexibility gained by letting go of face-contuinity in fact allows us to achieve slightly better locality bounds on some measures---see Table~\ref{tab:OmH}.

\begin{table}
\centering
\begin{tabular}{|l|r|rl|}
\hline
measure & best face-continuous & \multicolumn{2}{c|}{best vertex-continuous} \\
        & value                & value & some example curves \\
\hline
\WLMax  &                18.7 &       15.6 & B86.00.2c, B86.00.6c \\
\WLEuc  &                24.6 &       23.6 & B6.00.42, B88.00.81, B89.00.41 \\
\WLMan  &                90.5 & see $\leftarrow$ & A26.00.36 \\
\WS     &                1.74 & see $\leftarrow$ & A16.00.3c \\
\WBV    &                3.11 & see $\leftarrow$ & A26.00.00, A26.00.99 \\
\WBS    &                5.17 &       4.82 & B7.00.6e \\
\hline
\end{tabular}
\caption{Locality properties of the best order-preserving mono-Hilbert curves.}
\label{tab:OmH}
\end{table}

However, as soon as one is willing to give up on order-preservation, the face-continuous curves win again (see Table~\ref{tab:mH}), and no advantages of non-face-continuous curves were found.

\begin{table}
\centering
\begin{tabular}{|l|r|r|rl|}
\hline
measure & best order-pres.
        & best vertex-gated
        & \multicolumn{2}{c|}{best of all} \\
        & value
        & value
        & value & some example curves \\
\hline
\WLMax  &             15.6 &             12.4 & see $\leftarrow$ & A26.2b.b3 \\
\WLEuc  &             23.6 &             21.3 &             18.6 & F \\ 
\WLMan  &             90.5 &             89.7 & see $\leftarrow$ & many A26-curves; F has 89.8\\
\WS     &             1.74 & see $\leftarrow$ & see $\leftarrow$ & many A16-, A17-, A18-curves\\
\WBV    &             3.11 & see $\leftarrow$ & see $\leftarrow$ & A26.2b.b3, A26.00.00 \\
\WBS    &             4.82 &             4.40 &             4.05 & F \\
\hline
\end{tabular}
\caption{Locality properties of the best mono-Hilbert curves.}
\label{tab:mH}
\end{table}

\paragraph{Non-vertex-gated curves}
Hilbert curves are often said to start and end in a vertex, because this is what the two-dimensional Hilbert curve does. However, this may merely be a consequence of the fact that in two dimensions, there is only one mono-Hilbert curve and it happens to be vertex-gated. In higher dimensions it is possible to construct mono-Hilbert curves that are not vertex-gated, although in three dimensions, no \emph{order-preserving} non-vertex-gated mono-Hilbert curve exists. Non-order-preserving, non-vertex-gated three-dimensional mono-Hilbert curves do exist; most of these curves have bad locality properties (according to the measures discussed in Section~\ref{sec:measures}), but one of them is special. The face-gated type-F curve has the best \WLEuc-value and the best \WBS-value of all mono-Hilbert curves. It is one of the two most strongly bending mono-Hilbert curves---see below.

\paragraph{Non-order-preserving curves and bend size}
Allowing non-order-preserving curves results in improved measured locality on $\WLMax$, $\WLEuc$, and $\WBS$ (see Table~\ref{tab:mH}). This is due to two special curves in particular.

\begin{figure}
\centering
\hbox to \hsize{\hfill
\includegraphics[scale=0.8]{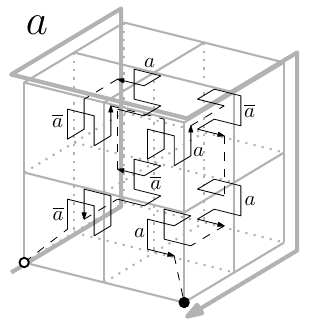}\hfill
\includegraphics[scale=0.8]{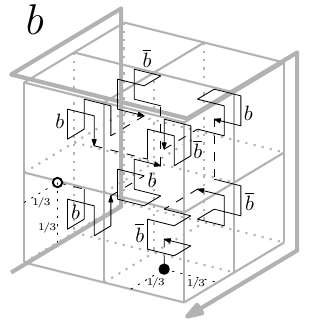}\hfill
}
\caption{%
(a) A26.0010\,1011.1011\,0011, a curve with optimal \WLMax and \WBV among mono-Hilbert curves.
(b) Curve F, the curve with the best-known \WLEuc-value among mono-Hilbert curves.}
\label{fig:reversing-curves}\label{fig:A26.2b.b3}
\end{figure}

Consider, at any level refinement, three cubic cells $a$, $b$ and $c$ of equal size that are consecutive in the scanning order. The size of the bend at $b$ could be expressed as the Euclidean distance between the centres of $a$ to $c$, divided by the width of one of these cells. If the curve is face-continuous, then all bends have size $\sqrt{2}$ or $2$; if the curve is only edge-continuous or vertex-continuous, then it must contain bends larger than $\sqrt{2}$ (consider a set of three consecutive cells that are octants of two larger consecutive cells that do not share a face).


There are only two mono-Hilbert curves that do not contain any bends larger $\sqrt{2}$: the curves A26.2b.b3 and F (Figure~\ref{fig:reversing-curves}). These two curves could thus be considered the most strongly bending mono-Hilbert curves, and it is indeed these two curves that hold the records for best measured locality properties on \WLMax, \WLEuc, \WLMan, \WBV and \WBS among mono-Hilbert curves. Neither of these curves is order-preserving.

\section{Poly-Hilbert curves}

The number of possible poly-Hilbert curves is unbounded, so we cannot enumerate all of them. The curves presented below are the result of a search by hand, and are not necessarily optimal.

\begin{figure}
\centering
(a)\includegraphics[scale=0.8]{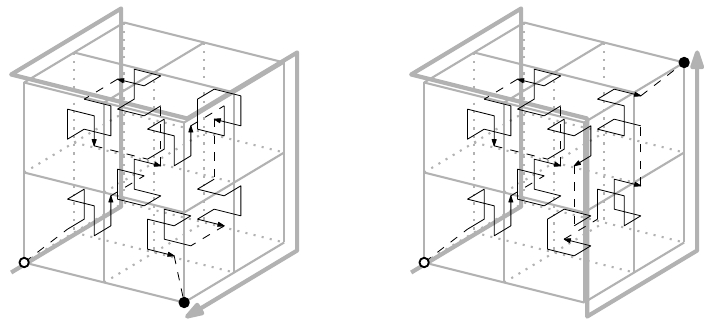}

\addvspace{2\baselineskip}
(b)\includegraphics[scale=0.8]{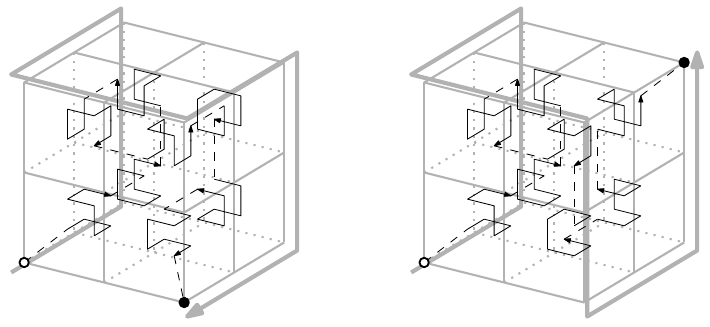}
\caption{%
(a) The Neptunus curve, a curve with optimal \WLMax, the best known \WBS, and within 1\% of the best known \WLEuc among poly-Hilbert curves.
(b) The Luna curve, a partially downward-compatible curve within 1\% of the best known \WLEuc, \WLMan, and \WS among poly-Hilbert curves.}
\label{fig:di-curves}
\end{figure}

\paragraph{$L_\infty$-locality and bounding-box surface ratio}
Figure~\ref{fig:di-curves}(a) shows my Neptunus curve, a poly-Hilbert curve of two rules with $\WLMax = 9.45$, $\WLEuc = 18.33$, and $\WBS = 3.53$. In the appendix we prove that the \WLMax-value is optimal for poly-Hilbert curves (Theorem~\ref{th:WLMaxlbd}). The \WBS-value is the best known among poly-Hilbert curves, and the \WLEuc-value is within 1\% of the best known among poly-Hilbert curves.
Like curves A26.2b.b3 and F, the Neptunus curve is maximally bending: it does not contain bends larger than $\sqrt{2}$.

\paragraph{$L_2$-locality, $L_1$-locality, and surface ratio}
Figure~\ref{fig:di-curves}(b) shows my Luna curve, a poly-Hilbert curve of two rules with $\WLEuc = 18.33, \WLMan = 75.60$, and $\WS = 1.70$. This is only slightly worse than the best poly-Hilbert curves with three rules I found, which have $\WLEuc = 18.11$ and $\WLMan = 75.39$, respectively. 
The Luna curve is maximally bending: it does not contain bends larger than $\sqrt{2}$. It is also partially downward-compatible: it has two-dimensional Hilbert orders on the top and the front face.

\begin{figure}
\centering
\includegraphics[width=\hsize]{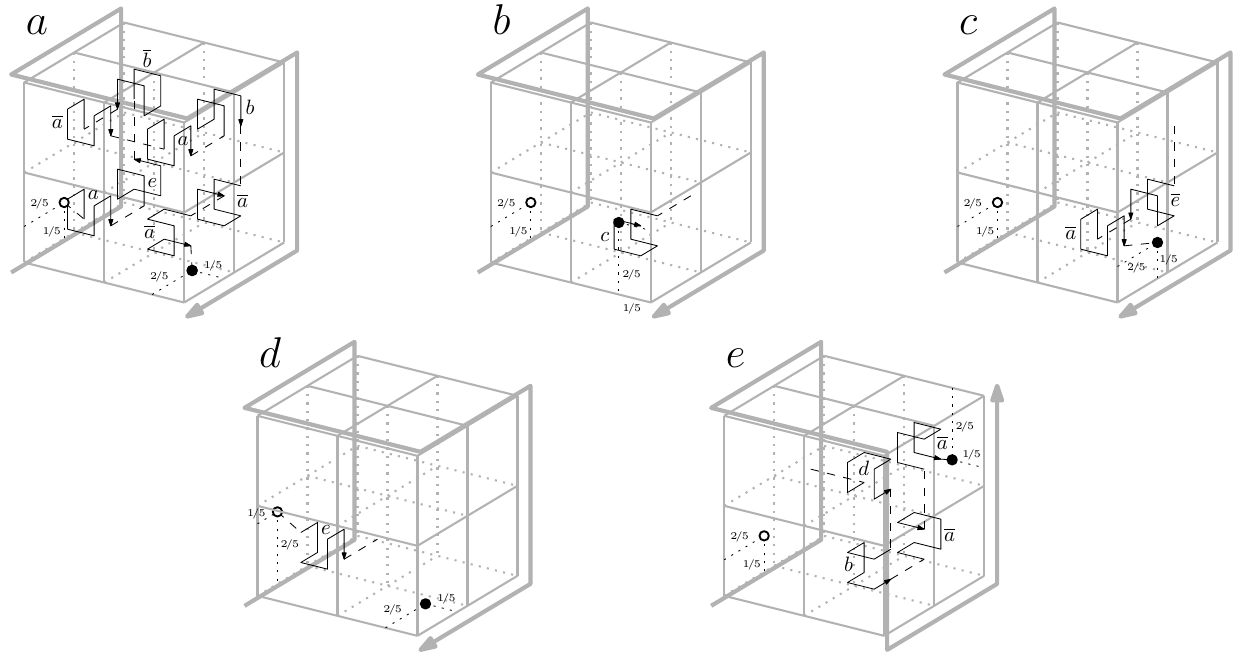}
\caption{The Iupiter curve: a curve with the best-known \WBV-value among poly-Hilbert curves. The octants that are left unspecified should be filled in as in rule $a$.}
\label{fig:jupiter}
\end{figure}

\paragraph{Bounding-box volume ratio}
The five-rule, face-continuous, face-gated Iupiter curve in Figure~\ref{fig:jupiter} is my current attempt at a scanning order with optimal worst-case bounding box volume ratio. With a value of 2.89 this ratio must be quite close to optimal: 
a lower bound of 2.76 for poly-Hilbert curves was found by exhaustively searching all octant-wise traversals of a cube down to the third level of recursion (a $8 \times 8 \times 8$ grid). The Iupiter curve was designed by deriving necessary structural properties for scanning orders with bounding box volume ratios below $84/29 \approx 2.90$ and $448/155 \approx 2.89$, respectively, and then designing a curve that has the first set of properties but, unfortunately, failing to guarantee the second set of properties.
Therefore, if I did not overlook some cases in the tedious manual design process, the bounding box volume ratio of this scanning order must actually be optimal for poly-Hilbert curves.

\section{Conclusions}

\begin{sidewaystable}\label{tab:results}
\caption{Properties and quality measures of space-filling curves discussed in this paper.}
\def\no{}
\def\half{\ensuremath{\circ}}
\def\yes{\ensuremath{\bullet}}

\addvspace\baselineskip
\def\arraystretch{1.25}
\begin{tabularx}{\hsize}{ll|@{\,}c@{\,}c@{\,}c@{\,}c@{\,}|@{ }c@{ }c@{ }r@{ }|@{\,}c@{\,}|@{\,}c@{\,}c@{\,}|X}
curve      &Fig.&O&D&S&C& \WLMax     & \WLEuc & \WLMan & \WS & \WBV & \WBS &
notes \\
\hline
A16.0000\,0000.0011\,1100 & \ref{fig:decomposable-curves} & \yes  &       & \yes  & \yes  & 25.2 & 27.1   & 112.0  & 1.74 & 3.71 & 5.79 &
opt \WS mH-curves
\\
A26.0000\,0000.0000\,0000 & \ref{fig:A26-curves}(b)       & \yes  & \no   & \half & \yes  & 24.2 & 26.2   & 98.3   & 1.80 & 3.11 & 5.56 &
opt \WBV mH-curves
\\
A26.0000\,0000.0011\,0110 &                               & \yes  & \no   & \no   & \yes  & 26.0 & 29.0   &  90.5  & 1.83 & 3.25 & 5.96 &
opt \WLMan OmH-curves
\\
A26.0000\,0000.1001\,1001 &                               & \yes  & \yes  & \half & \yes  & 24.2 & 26.4   &  99.6  & 1.80 & 3.11 & 5.56 &
opt \WBV mH-curves
\\
A26.0000\,0000.1101\,1011 & \ref{fig:A26-curves}(c)       & \yes  & \yes  & \half & \yes  & 28.0 & 28.9   & 99.6   & 1.80 & 3.50 & 6.43 &
2D-curves on five sides
\\
B\hphantom{0}7.0000\,0000.0110\,1110  &
                                                          & \yes  & \no   & \no   & \half & 15.7 & 24.6   & 113.9  & 1.95 & 4.31 & 4.82 &
opt \WBS OmH-curves
\\
B88.0000\,0000.1000\,0001 &
                                                          & \yes  & \no   & \half & \half & 16.6 & 23.6   & 111.2  & 2.06 & 4.20 & 4.95 &
opt \WLEuc OmH-curves
\\\hline
A26.0010\,1011.1011\,0011 & \ref{fig:reversing-curves}(a) & \half & \no   & \no   & \yes  & 12.4 & 22.9   &  99.6  & 1.80 & 3.11 & 4.40 &
opt \WBV and \WLMax mH-curves
\\
F          & \ref{fig:reversing-curves}(b)                & \half & \no   & \no   & \yes  & 14.0 & 18.6   & 89.8   & 1.80 & 3.14 & 4.05 &
opt \WBS, \WLEuc and near-opt \WLMan mH-curves
\\\hline
Neptunus   & \ref{fig:di-curves}(a)                       & 2     & \no   &  \no  & \yes  & 9.45 & 18.3   & 88.9   & 1.71 & 3.11 & 3.53 &
opt \WLMax, best known \WBS pH-curves
\\
Luna       & \ref{fig:di-curves}(b)                       & 2     & \hphantom{\half} &  \no  & \yes  & 14.0 & 18.3   & 75.6   & 1.70 & 3.11 & 4.05 &
best known \WS; near-best known \WLEuc, \WLMax pH-curves
\\
$\mathrm{H}^*$      & \cite{Chochia}                      & 10    & \hphantom{?} &  \no  & \yes  & 14.0 & 20.4   & 84.0   & 1.74 & 3.11 & 4.05 &
\\
Iupiter    & \ref{fig:jupiter}                            & 5     & \no   & \no   & \yes  & 17.0 & 24.9   & 88.7   & 1.76 & 2.89 & 4.92 &
best known \WBV pH-curves
\\\hline
(no curve) &       &       &       &       &       & 8.25       & 11.1   & 42.625 &        &      &      &
lower bounds cube-filling curves
\end{tabularx}

\addvspace\baselineskip
O: \half\ simple; \yes\ order-preserving; number: number of rules of poly-Hilbert curve;\\
D: 
   \yes\ downward-compatible;\\
S: \half\ symmetric; \yes\ symmetric and binary-decomposable;\\
C: \half\ edge-continuous; \yes\ face-continuous;\\
Quality measures have been computed to (at least) the indicated precision.
\end{sidewaystable}

Hilbert curves can be generalized to three dimensions in multiple ways. In this paper we have seen that interesting generalized Hilbert curves exist that had not been considered in the literature before, including a curve that traverses most faces of the unit cube in the order of two-dimensional Hilbert curves.

We saw curves that are not face-continuous, curves that are not vertex-gated, and curves that use reversed and/or multiple recursive rules. Especially the use of multiple recursive rules, sometimes in combination with gates on faces, results in improved quality measures---see Table~\ref{tab:results} for an overview. Compared to symmetric, non-reflecting, face-continuous, order-preserving mono-Hilbert curves, poly-Hilbert curves with two recursive rules can bring down $\WLMax$ from 24.2 to 9.45, $\WLEuc$ from 26.2 to 18.3, and $\WBS$ from 5.56 to 3.53. All of these results are due to maximally bending curves: curves without bends larger than $\sqrt{2}$. Our calculations on this extended set of curves have reduced the gaps between lower bounds and upper bounds on the locality measures of the best cube-filling curves:
the gap for $\WLMax$ by 94\%, 
for $\WLEuc$ by 68\%,         
and for $\WLMan$ by 41\%.     
For $\WS$, $\WBV$ and $\WBS$ we now have upper bounds of 1.70, 2.89 and 3.53, respectively.

In fact, we saw that we have found a poly-Hilbert curve of which the $L_\infty$-locality (9.45) is optimal among all poly-Hilbert curves. Unfortunately, proofs of optimality or near-optimality for $\WLEuc$ and $\WLMan$ seem much harder to get; the necessary case analysis quickly becomes too complicated to be informative, especially when curves that are not vertex-gated are taken into account.



\paragraph*{Acknowledgements}
Many computations of the locality measures reported in this paper were made possible by the work of Simon Sasburg~\cite{Sasburg}. His improvements of the algorithm from Haverkort and Van Walderveen resulted in the speed and accuracy needed to analyse the locality and bounding-box quality measures of millions of curves. Moreover, he designed and implemented the computation of the surface ratios.

\bibliographystyle{abbrv}

\break
\appendix

\section{Possible locations of gates in mono-Hilbert curves}

The following lemmas hold for all three-dimensional mono-Hilbert curves, regardless whether they are order-preserving or not.

\begin{lemma}\label{lem:vertex-vertex-gates}
For gates at vertices, we may assume that the entrance gate is at $(0,0,0)$ and the exit gate is either at $(1,0,0)$ or at $(0,1,1)$.
\end{lemma}
\begin{proof}
The unit cube can always be rotated and/or reflected so that the above assumptions hold, unless the entrance and exit gates are at vertices of the unit cube that are opposite of each other with respect to the centrepoint of the cube. However, in that case the gate from the first to the second octant would be at the centrepoint of the unit cube, and the exit gate of the second octant, no matter which octant it is, would be at a vertex of the unit cube, where no connection to a third octant would be possible. Hence no connection scheme with gates at opposite vertices could be completed.
\end{proof}

\begin{lemma}\label{lem:edge-edge-gates}
For gates in the interior of edges, we may assume that the entrance gate is at $(\frac13,0,0)$ and the exit gate is at $(1,\frac13,1)$.
\end{lemma}
\begin{proof}
Assume the unit cube has been rotated and/or reflected so that the entrance gate lies on the lower front edge, and the lower left front octant is the first octant in the scanning order. Let $a$ be the distance of the entrance gate to the lower left front vertex of the unit cube, and let $z$ be the distance of the exit gate to the closest vertex of the unit cube. Note that we must have $0 < a \leq 1/2$ and $0 < z \leq 1/2$. Let the gates of the octants be $g_0,...,g_8$, where the entrance gate of the $i$-th octant in the order is $g_{i-1}$, and its exit gate is $g_i$. Let $p_i$ be the vertex of the unit cube or one of its octants that lies closest to $g_i$.

We will first prove that $a = z$. To prove this, assume, for the sake of contradiction, that $a \neq z$. Then the entrance and exit gates of the octants can only be matched up if the order is reversed either (i) in every odd or (ii) in every even octant. We first analyse the second case. In this case, gate $g_0$ is the entrance gate of the unit cube, at distance $a$ from a vertex of the unit cube, and it is the entrance gate of the first octant, so it must be at distance $a/2$ from one of the vertices of the first octant. Therefore we must have either $a = a/2$ (which has no solutions satisfying $a > 0$) or $a = 1/2 - a/2$, which solves to $a = 1/3$. The last octant is reversed, so gate $g_8$ was originally an entrance gate, and this gate is the exit gate of the unit cube. Therefore we must have $z = \frac12 - \frac12a = 1/3$ or $z = \frac12a = 1/6$. Under the assumption $a \neq z$, the latter must be the case.
Now consider the length of a shortest path $\pi$ from $g_0$ to $p_8$ along the edges of the unit cube and its octants. The length of $\pi$ is $k/3$, for some $k \in \Naturals$. From $g_0$ to $p_1$ we could follow a scaled-down version of $\pi$, of length $k/6$. From $p_1$ to $g_2$ takes another $k/6$. We could continue our path in a similar manner to $p_3$, $g_4$, $p_5$, $g_6$, $p_7$ and $g_8$, and finally add a section of length $z$ to get to $p_8$. The total length of this path from $g_0$ to $p_8$ is $(8k+1)/6$. If this is not the shortest path from $g_0$ to $p_8$, we can shorten it by removing cycles and rerouting sections of it, but these operations will maintain that the path has length $m/6$ for an \emph{odd} number $m$. This contradicts the observation that $\pi$ has length $k/3$, that is, $m/6$ for an \emph{even} number $m$. Hence we cannot have $a \neq z$ and we must have $a = z$. In case (i), when every odd octant is reversed, we can argue in a symmetric way, following the curve from gate $g_8$ to $g_0$ instead of from $g_0$ to $g_8$, again finding that we must have $a=z$.

Now, with $a = z$, to match the entrance gate of the unit cube to the entrance gate of its first octant, we must have either $a = a/2$ (which has no solutions satisfying $a > 0$) or $a = 1/2 - a/2$, that is, $a = 1/3$. Thus we get $a = z = 1/3$, and the entrance gate of the unit cube is at $(\frac13,0,0)$. The exit gate cannot be in the same octant, and by similar arguments as in the proof of Lemma~\ref{lem:vertex-vertex-gates}, the exit gate cannot be in the opposite octant with respect to the centrepoint of the unit cube. The exit gate cannot lie on one of the edges that share vertex $(0,0,0)$, because then the $L_\infty$-distance between the entrance and the exit gate would be at most $2/3$; the $L_\infty$-distance between $g_7$ and $g_8$ would thus have to be half of that (at most $1/3$), but the $L_\infty$-distance between the seventh octant and $g_8$ would be $1/2$: a contradiction. By a symmetric argument, the exit gate cannot lie at $(1,\frac13,0)$ or $(1,0,\frac13)$; one would get a contradiction regarding the distance between $g_0$ and the second octant. The exit gate cannot lie on an edge parallel to the edge that contains the entrance gate, because then the curve could never make a connection from the left half to the right half of the unit cube. Note that none of the gates $g_1,...,g_7$ can lie on an edge of the unit cube, because there they could not be an exit gate of one octant and an entrance gate of another. Therefore the exit gate cannot lie at $(1,\frac23,0)$ or $(1,0,\frac23)$: $g_1$ would then lie at $(0,\frac13,0)$ (modulo swapping $y$ and $z$), which is on an edge of the unit cube. The exit gate cannot lie at $(0,\frac13,1)$ or $(0,1,\frac13)$, because then $g_1$ would have to lie at $(\frac12,\frac12,\frac16)$ (modulo swapping the $y$- and $z$-coordinates) and $g_2$ would have to lie at $(\frac12\pm\frac16, \frac12\pm\frac12,0)$ (modulo swapping the $x$- and $y$-coordinates), which is on an edge of the unit cube. With the entrance gate at $(\frac13,0,0)$, the only remaining places for the exit gate are therefore at $(1,\frac13,1)$, $(1,1,\frac13)$, $(0,\frac23,1)$ or $(0,\frac23,1)$; these cases are all the same modulo rotation, reflection and/or reversal.
\end{proof}

\begin{lemma}\label{lem:face-face-gates}
For gates in the interior of faces, we may assume that the entrance gate is at $(0,\frac13,\frac13)$ and the exit gate is at $(\frac23,\frac13,0)$.
\end{lemma}
\begin{proof}
We argue in a similar manner as for gates on edges. Let $a_1 > 0$ and $a_2 > 0$ be the distance of the entrance gate to the closest edge and to the second-closest edge, respectively. Similarly let $z_1 > 0$ and $z_2 > 0$ be the distance of the exit gate to the closest edge and to the second-closest edge, respectively. Note that we have $a_1, a_2, z_1, z_2 < \frac12$, otherwise gates would have to lie on edges of octants.

We first make some general observations. Because the gates between octants are in the interior of faces, each octant must share a face with the previous one, and so each octant differs from the previous in only one coordinate. Therefore the last octant differs from the first one in an odd number of coordinates, that is, the last octant shares a face with the first one, or it is opposite to the first octant with respect to the centrepoint of the unit cube. The latter case cannot happen, by a similar argument as in the proof of Lemma~\ref{lem:vertex-vertex-gates}, so the last octant must share a face with the first octant. Furthermore, the exit gate cannot lie on a face parallel to the face that contains the entrance gate, because then $g_2$ would have to lie on the same face as $g_1$, or on the outside of the unit cube opposite to the face containing $g_0$: in both cases $g_2$ could not connect to any other octant than the first and the second.

We will now show that $(a_1,a_2) = (z_1,z_2)$. For the sake of contradiction, suppose $(a_1,a_2) \neq (z_1,z_2)$. Then we can again argue that the entrance gate of the unit cube is also the entrance gate of the first octant, and therefore we must have either:\begin{itemize}
\item $a_1 \in \{\frac12 a_1, \frac12 - \frac12 a_1\}$ and $a_2 \in \{\frac12 a_2, \frac12 - \frac12 a_2\}$, which solves to $a_1 = a_2 = 1/3$;
\item or $a_1 \in \{\frac12 a_2, \frac12 - \frac12 a_2\}$ and $a_2 \in \{\frac12 a_1, \frac12 - \frac12 a_1\}$, which solves to $a_1 = 1/5$, $a_2 = 2/5$.
\end{itemize}
Now observe that $g_1$ is the exit gate of the unit cube as it is scaled down and possibly rotated and reflected to fit in the first octant. At the same time it is the exit gate of the unit cube that is scaled down, reversed, and possibly rotated and reflected to fit in the \emph{second} octant. If $z_1 \neq z_2$, the rotations and reflections can only match in one way, so that $g_2$ must be on a face of the second octant parallel to the face of the first octant that contains $g_0$. Continuing this reasoning, we would find that $g_0$, $g_2$, $g_4$, $g_6$ and $g_8$ are all on parallel faces. But if $g_0$ and $g_8$ are on parallel faces, then we could only connect octants in a single row and never connect eight octants in a block of $2\times2\times2$. Therefore, if $(a_1,a_2) \neq (z_1,z_2)$, we must have $z_1 = z_2$. Since the exit gate of the unit cube is also the exit gate of the last octant (which is really a reversed entrance gate), we now must have $z_1 = z_2 \in \{\frac12 a_1, \frac12 - \frac12 a_1\} \cap \{\frac12 a_2, \frac12 - \frac12 a_2\}$. This only has a solution if $a_1 = a_2 = 1/3$; then $z_1 = z_2 = 1/6$. Without loss of generality, assume the entrance gate $g_0$ is at $(0,\frac13,\frac13)$. By the observations made above, $g_1$ must now lie on a suboctant of the first octant which is face-adjacent to the suboctant that contains $g_0$, but not on a face of the first octant opposite to $g_0$; the gate $g_1$ must lie at a distance of $z/2 = 1/12$ from the edges of the first octant. This leaves four possible positions for $g_1$: $(\frac1{12},\frac12,\frac1{12})$ and $(\frac5{12},\frac12,\frac5{12})$, and the same with $y$- and $z$-coordinates swapped. One may now verify that in both cases, the exit point of the second octant would now have to lie on the outside of the unit cube, where no connection to the third octant is possible. Hence we must have $(a_1,a_2) = (z_1,z_2)$.

Now, with $(a_1,a_2) = (z_1,z_2)$, to match the entrance gate of the unit cube to the entrance gate of its first octant, we must have either $(a_1,a_2) = (\frac15,\frac25)$ or $(a_1,a_2) = (\frac13,\frac13)$. We will now analyse these cases one by one.

The case of $(a_1,a_2) = (\frac15,\frac25)$: assume without loss of generality that $g_0$ is at $(0,\frac25, \frac15)$, which is in the lower left back suboctant of the lower left front octant. By the observations made above, $g_1$ must now be on the back face of the upper left back suboctant, the back face of the lower right back suboctant, or the top face of the upper left back suboctant. More precisely, $g_1$ must be one of: $(\frac3{10},\frac12,\frac1{10})$, $(\frac25,\frac12,\frac15)$, $(\frac1{10},\frac12,\frac3{10})$, $(\frac15,\frac12\,\frac25)$, $(\frac1{10},\frac3{10},\frac12)$, $(\frac15,\frac25,\frac12)$. None of these cases work out: in the first, third and fifth case, $g_2$ is forced onto the outside of the unit cube; in the second case, the sequence of gates $g_0, g_1, ...$ is forced to remain in the plane $z = \frac15$ and the upper half of the unit cube cannot be reached; in the sixth case, the sequence of gates $g_0, g_1, ...$ is forced to remain in the plane $y = \frac25$ and the back half of the unit cube cannot be reached; in the fourth case, if the first octant is not reversed, the transformation that maps the first octant and its entrance gate to the unit cube and its entrance gate, will map the exit gate $g_1$ at $(\frac15,\frac12,\frac25)$ to $g_8$ at $(\frac25,\frac45,0)$, but then the curve would have to visit the lower left back octant twice (once reaching it through $g_1$; once leaving it through $g_8$); if, in the fourth case, the first octant is reversed, the transformation that maps the first octant and its exit gate to the unit cube and its entrance gate, will map the entrance gate $g_0$ at $(0,\frac25,\frac15)$ to $g_8$ in the upper left front octant at $(\frac15,0,\frac35)$; however, the only way of connecting transformations of the first octant up such that the traversal ends in the upper left front octant, leads over $g_2 = (\frac25,\frac7{10},\frac12)$, $g_3 = (\frac12,\frac9{10},\frac7{10})$, $g_4 = (\frac35,\frac7{10},\frac12)$, $g_5 = (\frac45, \frac12, \frac25)$, $g_6 = (\frac35,\frac3{10},\frac12)$ and $g_7 = (\frac12,\frac1{10},\frac7{10})$ to $g_8$ at $(\frac3{10},0,\frac9{10})$, contradicting $g_8 = (\frac15,0,\frac35)$.

The case of $(a_1,a_2) = (\frac13,\frac13)$: assume without loss of generality that $g_0$ is at $(0,\frac13,\frac13)$. By the observations made above, $g_1$ must be on the top face or the back face of the lower left octant; these are equivalent modulo a swap of the $x$- and $y$-coordinates, so we may assume $g_1$ is on the back face of the lower left octant; more precisely $g_1$ must be at $(\frac16,\frac12,\frac16)$ or $(\frac13,\frac12,\frac13)$. In the first case, the $L_\infty$-distance between the entrance and exit gates of an octant is only $\frac16$, and thus $g_2$ cannot lie on the boundary of any octant other than the first and the second; any possible third octant is out of reach. Hence $g_1$ must be at $(\frac13,\frac12,\frac13)$. Considering the possible transformations of the first octant (similar to our analysis of the fourth case with $(a_1,a_2) = (\frac15,\frac25)$), we now find that $g_8$ must be at $(\frac23,\frac13,0)$ or $(\frac23,0,\frac13)$; by swapping the $y$- and $z$-coordinates we can ensure that $g_0$ is at $(0,\frac13,\frac13)$ and $g_8$ is at $(\frac23,\frac13,0)$.
\end{proof}

\begin{lemma}\label{lem:vertex-edge-gates}
For one gate at a vertex and one gate in the interior of an edge, we may assume that the entrance gate is at $(0,0,0)$ and the exit gate is at $(1,\frac12,0)$.
\end{lemma}
\begin{proof}
Without loss of generality, assume the entrance gate is at a vertex and the exit gate is in the interior of an edge. This implies that the rule must be reversed in every second octant to be able to match the gates between the octants. Thus the traversal ends at a \emph{vertex} of the eighth octant that lies in the interior of an edge of the unit cube---thus it actually lies in the middle of that edge. The exit gate cannot lie on an edge incident on the entrance gate, because then, after traversing the first octant within a cube, it would be impossible to connect to the second octant. The exit gate cannot lie at $(1,\frac12,1)$: the sequences of gates $g_0, g_1, ...$ that can be made by composing transformations of the unit cube with its gates are fairly limited and it is easy to verify by hand that they cannot be completed to end at $(1,\frac12,1)$. Therefore the exit gate must lie at $(1,\frac12,0)$ (modulo reflections and rotations).
\end{proof}

\begin{lemma}\label{lem:vertex-face-gates}
For one gate at a vertex and one gate in the interior of a face, we may assume that the entrance is at $(0,0,0)$ and the exit gate is at $(1,\frac12,\frac12)$.
\end{lemma}
\begin{proof}
Without loss of generality, assume the entrance gate is at a vertex and the exit gate is in the interior of a face. By a similar reasoning as in Lemma~\ref{lem:vertex-edge-gates}, this implies that the exit gate lies right in the middle of that face. The exit gate cannot lie in a face incident on the entrance gate, because then, after traversing the first octant within a cube, it would be impossible to connect to the second octant. Therefore the exit gate must lie at $(1,\frac12,\frac12)$ (modulo reflections and rotations).
\end{proof}

\begin{lemma}\label{lem:edge-face-gates}
There is no generalized simple Hilbert curve with one gate in the interior of an edge and one gate in the interior of a face.
\end{lemma}
\begin{proof}
Assume the entrance gate is in the interior of an edge; the exit gate is in the interior of a face. This implies that the rule must be reversed in every second octant, and we end on a centreline of the face of the unit cube, but not in the centrepoint of the face. Without loss of generality, assume we end on the back face, in the top half of the vertical centreline. Then the first and the second octant must be connected back to back, possibly reflected but with matching rotations. Thus $g_0$ and $g_2$ must lie on parallel edges. Continuing this reasoning, we would find that $g_0$, $g_2$, $g_4$, $g_6$ and $g_8$ are all on parallel edges, hence $g_0$, the entrance gate of the unit cube, is on a vertical edge, parallel with the face that contains the exit gate. But then, tracing the sequence of octants back from the end, we find that we can never reach the horizontal faces and edges on which we could connect to the lower octants. Therefore it is impossible to realize a connection scheme with one gate in the interior of an edge and the other gate in the interior of a face.
\end{proof}

\section{Proving \WLMax-optimality of the Neptunus curve among poly-Hilbert curves}

\begin{theorem}\label{th:WLMaxlbd}
No poly-Hilbert curve has $\WLMax < 189/20 = 9.45$.
\end{theorem}
\begin{proof}
To prove the theorem, we will analyse what properties a curve with $\WLMax \leq 9.45$ should have.

First observe that, after entering any cell, the curve traverses at most $1/2 + 1/8\cdot 1/2 + 1/8^2 \cdot 1/2 + ... = 4/7$ of it before having hit all sides of the cell. Hence, the bounding box of any three consecutive cells of the grid, at any level of recursion, should have size at most twice a cell's width in each dimension (we call this the \emph{always-turn} rule), otherwise the curve would be subject to a \WLMax lower bound of $3^3 / (4/7 + 1 + 4/7) = 63 / 5 = 12.6$ (see Figure~\ref{fig:WLMaxlbd}(a)).

\begin{figure}
\centering
\includegraphics[width=\hsize]{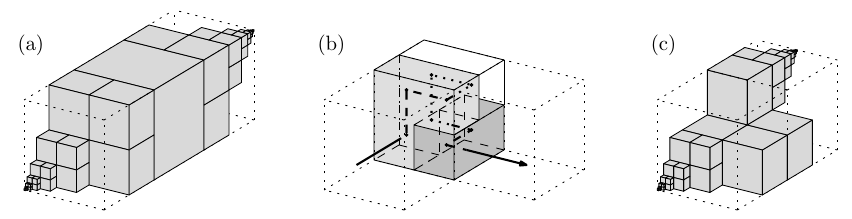}
\caption{%
(a) If the bounding box of three consecutive cells is as long as three times a cell's width, then a $L_\infty$-distance of 3 cell widths would be realized by a curve section filling at most 4/7 + 1 + 4/7 cells.\quad
(b) A cell $C_i$ with its predecessor and successor in the scanning order, the octants that constitute its head, and the octants that constitute its tail.\quad
(c) If one cell's head and the next cell's tail occupy only five octants, a $L_\infty$-distance of 4 octant widths would be realized by a curve section filling at most 4/7 + 5 + 4/7 octants.
}
\label{fig:WLMaxlbd}
\end{figure}

Consider a curve traversing the cells $C_1,...,C_k$ of a grid of $k$ cubes of equal size. For $1 \leq i < k$, let $h_i$ be an axis-parallel plane that separates $C_i$ from $C_{i+1}$. Define $\tail_{h_i}(C_i)$ as the set of octants of $C_i$ that appear along the curve before the first octant that does not touch $h_{i-1}$. Symmetrically, define $\head_{h_i}(C_i)$ as the set of octants of $C_i$ that appear along the curve after the last octant that does not touch $h_i$. (When it is not relevant which plane is taken for $h_i$, or when it is unique, we may omit the subscripts.) Note that as a result of the \emph{always-turn} rule, $h_{i+1}$ can never be parallel to $h_i$. This implies that every cell $C_i$ contains at least two octants that touch neither $h_i$ nor $h_{i+1}$; hence $|\tail(C_i)| + |\head(C_i)| \leq 6$ (see Figure~\ref{fig:WLMaxlbd}(b)).

For $\WLMax \leq 9.45$ we need $|\head(C_i)| + |\tail(C_{i+1})| \geq 6$ for all $1 \leq i < k$ (we call this the \emph{stick-six} rule), because if $|\head(C_i)| + |\tail(C_{i+1})| \leq 5$, the section of the curve that traverses the octant of $C_i$ before $\head(C_i)$, the octants in $\head(C_i)$ and $\tail(C_{i+1})$, and the octant of $C_{i+1}$ after $\tail(C_{i+1})$, results in a $\WLMax$ lower bound of $4^3 / (4/7 + 5 + 4/7) = 448 / 43 \approx 10.42$ (see Figure~\ref{fig:WLMaxlbd}(c)). Together with the observation that $|\tail(C_i)| + |\head(C_i)| \leq 6$, this implies that there is a section of the curve (we call this a \emph{regular} section) where for each level of recursion it holds that either all heads contain two octants and all tails contain four, all heads and all tails contain three octants, or all heads contain four octants and all tails contain two.

We will now argue that in a regular section, consecutive cells always share a face. For the sake of contradiction, suppose there are two consecutive cells $C_i$ and $C_{i+1}$ that do not share a face. Let $f$ and $g$ be two axis-parallel planes that separate these cells. Note that the conditions on the sizes of heads and tails derived above must hold regardless of whether $f$ or $g$ is taken to be $h_i$. Assume we have $|\head_f(C_i)| \geq 3$ (if not, the proof goes through in a symmetric way, starting from $|\tail_f(C_{i+1})| \geq 3$)
. Then at least one octant of $\head_f(C_i)$ is not adjacent to $g$, so we have $|\head_g(C_i)| = 2$ and $|\tail_g(C_{i+1})| = 4$; then exactly two octants of $\tail_g(C_{i+1})$ are adjacent to $f$ so we get $|\tail_f(C_{i+1}| = 2$ and $|\head_f(C_i)| = 4$. Because $|\head_f(C_i)| = 4$ we have $|\tail(C_i)| = 2$ (regardless of $h_{i-1}$), and because $|\tail_g(C_{i+1})| = 4$ we have $|\head(C_{i+1})| = 2$ (regardless of $h_{i+1}$). But in a regular section, either the heads or the tails must be bigger than two octants. So in regular sections, consecutive cells always share a face.

\begin{figure}
\centering
\includegraphics[width=\hsize]{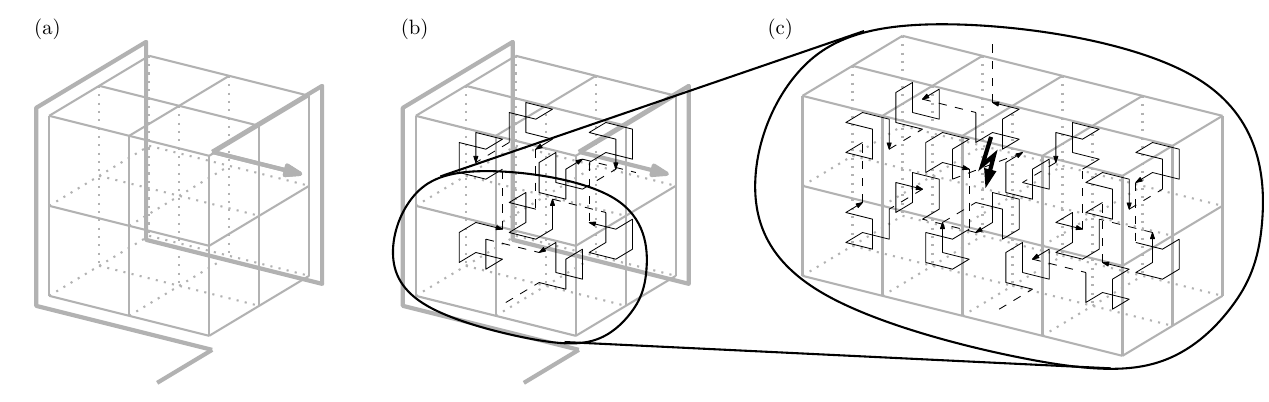}
\caption{%
(a) The base pattern for heads and tails of three octants.\quad
(b) The only way in which eight of these patterns can be combined to traverse a block of eight cells.\quad
(c) Combining these traverals recursively would result in mismatched gates.}
\label{fig:WLMaxlbd2}
\end{figure}

Consider composing a traversal of a grid of $2^3$ cells, with $4^3$ octants in total, such that in each cell, all heads and tails contain three octants. Then the traversal of the octants in each cell must be an appropriately rotated and/or reflected version of the pattern shown in Figure~\ref{fig:WLMaxlbd2}(a), and such patterns can be composed to traverse all $4^3$ octants in only one way (modulo rotations and reflections), see Figure~\ref{fig:WLMaxlbd2}(b), resulting in a larger pattern of which the head and the tail contain three octants. However, these cannot be combined again into a traversal of a grid of size $8^3$: at the marked spot in Figure~\ref{fig:WLMaxlbd2}(c), consecutive cells would not share a face. Hence, from the second level of recursion down, no regular section can contain heads and tails of three octants, and either the heads or the tails contain only two octants.

\begin{figure}
\centering
\includegraphics[width=0.8\hsize]{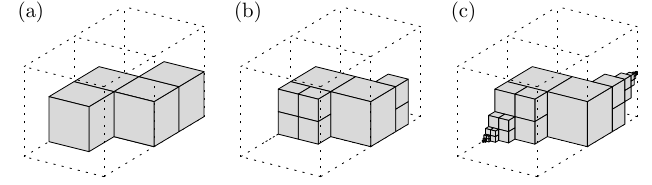}
\caption{An example of the unavoidable configuration that results in a \WLMax lower bound of 9.45.}
\label{fig:WLMaxlbd3}
\end{figure}

Now consider such a head (or tail) of two octants, together with the previous and the next octant in the scanning order (Figure~\ref{fig:WLMaxlbd3}(a)). In such a section of the curve, a $L_\infty$-distance of 3 is realized by the two central octants, plus the head of the first octant and the tail of the second octant, together occupying six suboctants (see Figure~\ref{fig:WLMaxlbd3}(b)), plus the preceding head and the following tail within the suboctants (See Figure~\ref{fig:WLMaxlbd3}(c)), etc. This adds up to $2 + 6/8^1 + 6/8^2 + ... = 20/7$ octants. Thus $\WLMax$ must be at least $3^3 / (20/7) = 189/20 = 9.45$, and no poly-Hilbert curve can have $\WLMax < 189/20 = 9.45$.
\end{proof}

\clearpage

\section{Update and erratum}\label{apx:update}

This manuscript discusses space-filling curves that recursively traverse a cube octant by octant. If the curve within each octant is similar (by scaling, translation, reflection, rotation and/or reversal) to the complete curve, then we call the curve self-similar. This manuscript discusses both self-similar and non-self-similar curves.

\subsection*{If you are interested in self-similar curves}
You may prefer my more recent manuscript ``How many three-dimensional Hilbert curves are there?'' (\url{http://arxiv.org/abs/1610.00155}). The new manuscript:\begin{itemize}
\item offers a more extensive account of properties that may characterize octant-wise self-similar cube-filling curves, supported by references to recent and older literature;
\item supports more observations by proofs, including a proof of the minimal defining properties of the two-dimensional Hilbert curve;
\item contains a much more structured naming scheme which allows us to recognize more symmetries and similarities between the curves;
\item contains a full list of octant-wise self-similar cube-filling curves (made possible by the new naming scheme, which eliminates the need for thousands of figures);
\item highlights and illustrates many more interesting curves;
\item comes with software;
\item includes a limited discussion of higher-dimensional curves;
\item is submitted to a peer-reviewed journal.
\end{itemize}
Some of the concepts defined in the old manuscript (the one you are reading) have been reconsidered, adapted and/or renamed in the new manuscript, mostly as a result of insights developing during more recent work. In particular:\begin{itemize}
\item the definition of \emph{symmetry} has been made more precise;
\item \emph{binary decomposability} has been replaced by the slightly less strict concept of \emph{metasymmetry};
\item \emph{downward compatibility} has been replaced by more general definitions of \emph{harmony};
\item \emph{vertex-continuity} is simply called \emph{continuity} in the new manuscript;
\item \emph{edge-continuity} is not discussed in the new manuscript;
\item the notion of being \emph{non-reflecting} is not discussed in the new manuscript;
\item \emph{bend size} has been replaced by \emph{hyperorthogonality}.
\end{itemize}
Likewise, in the new manuscript I have chosen different examples and have not reproduced all numbers reported in the tables in the old manuscript. If you see any value in any of the old definitions, numbers, or examples that have been omitted or modified in the new manuscript, please let me know and I will consider including them in the new manuscript and giving them a proper treatment after all.

In any case I advice against using the old manuscript's curve naming scheme (which was not even completely described). The new manuscript describes a much better way to give names to the curves.

\subsection*{If you are interested in non-self-similar curves (poly-Hilbert curves)}
The old manuscript (the one you are reading) is still relevant. You may want to consult the new manuscript for the better background, notation and definitions, but the new manuscript does not include any results on poly-Hilbert curves. For the Neptunus, Luna and Iupiter curves and the proof that the first has optimal $L_\infty$-dilation, the document you are currently reading is still the only source.

\subsection*{Erratum}
On page~7, I claimed that for $i = \infty$ and $i = 1$, the $L_i$-diameter ratio and the $L_i$-bounding ball ratio are always equal, modulo a fixed constant factor, and I based this on a claim in Footnote~6, that the diameter of the minimum bounding $L_i$-ball of any set $S$ is equal to the $L_i$-diameter of $S$. For $i = \infty$ the footnote is correct but for $i = 1$ it is not, at least not in three or more dimensions (consider the set of points $(0,0,0)$, $(0,1,1)$, $(1,0,1)$ and $(1,1,0)$). Consequently I have no proof for my claim that the $L_1$-diameter ratio and the $L_1$-bounding ball ratio are always equal modulo a fixed constant factor.

\addvspace\baselineskip
\raggedleft Herman Haverkort, 1 October 2016

\end{document}